\documentclass[12pt]{article}
\usepackage[a4paper,left=1.0in,right=1.0in,top=1.1in,bottom=1.3in]{geometry}

\usepackage{microtype}

\usepackage{authblk}
\usepackage{xcolor}
\usepackage{etex,ifthen}
\usepackage{etoolbox}
\usepackage{graphicx}
\usepackage{tikz}
\usepackage{pgfplots}
\pgfplotsset{compat=1.14}
\usepackage{tikz-cd}
\usepackage{url}
\usepackage{calc}
\usepackage[font={small,it}]{caption}
\usepackage{complexity}
\usepackage{alphalph}
\usepackage{amsthm}
\usepackage{thmtools}
\usepackage{eqparbox}
\usepackage{afterpage}
\usepackage{amsfonts,amsmath,amssymb}
\usepackage{algorithm}
\usepackage{algpseudocode}
\usepackage{bm}
\usepackage{framed}
\usepackage{listings}
\usepackage{hyperref}
\usepackage{verbatim}
\usepackage[all]{xy}

\usetikzlibrary{chains,fit}
\usetikzlibrary{shapes.geometric}
\usetikzlibrary{matrix,positioning,calc}
\usetikzlibrary{decorations.markings}
\usetikzlibrary{decorations.pathmorphing}
\tikzstyle{vertex}=[circle, draw, inner sep=0pt, minimum size=6pt]
\tikzstyle{vertbox}=[draw, inner sep=0pt, minimum size=8pt]
\newcommand{\vertex}{\node[vertex]}
\newcommand{\vertbox}{\node[vertbox]}

\newcommand{\oset}[3][0ex]{%
	\mathrel{\mathop{#3}\limits^{
			\vbox to#1{\kern-2\ex@
				\hbox{$\scriptstyle#2$}\vss}}}}

\newcommand{\lefty}{\mathsf{left}}
\newcommand{\righty}{\mathsf{right}}

\newcommand{\rev}{\mathbf{rev}_\gamma}
\newcommand{\revl}{\mathbf{rev}_{\gamma-\ell}}
\newcommand{\nbhd}{\mathsf{nbhd}}

\newcommand{\rect}{G^{\mathrm{bit}}_k}
\newcommand{\recth}{H^{\mathrm{bit}}_k}
\newcommand{\Ghard}{G^{\mathrm{hard}}_k}
\newcommand{\spec}{\mathbf{spcl}}
\newcommand{\Gzero}{G_{\mathrm{zero}}}
\newcommand{\Gset}{G_{\mathrm{set}}}

\newcommand{\cG}{\mathcal{G}}
\newcommand{\cH}{\mathcal{H}}
\newcommand{\cI}{\mathcal{I}}
\newcommand{\cJ}{\mathcal{J}}
\newcommand{\cF}{\mathcal{F}}
\newcommand{\cP}{\mathcal{P}}
\newcommand{\cS}{\mathcal{S}}
\newcommand{\cU}{\mathcal{U}}

\newcommand{\intgraph}{G_{\mathrm{int}}}
\newcommand{\intgraphh}{H_{\mathrm{int}}}
\newcommand{\newgraph}{H_{\mathrm{MH}}}
\newcommand{\Eslant}{E_{\mathrm{slant}}}
\newcommand{\bv}{\mathsf{B}}
\newcommand{\nbv}{\deg_{\geq 3}}

\newcommand{\greedy}{P_G^{\text{gr}}}
\newcommand{\greedyH}{P_{H_1}^{\text{gr}}}
\newcommand{\Gleft}{G_{\text{left}}}
\newcommand{\Gright}{G_{\text{right}}}
\newcommand{\Hleft}{H_{\text{left}}}
\newcommand{\Hright}{H_{\text{right}}}
\newcommand{\Pleft}{P_{\text{left}}}

\newcommand{\lca}{\mathbf{lca}} 
\newcommand{\floor}[1]{\left\lfloor #1 \right\rfloor}

\newcommand{\upvert}{E_{\mathrm{up}}}
\newcommand{\downvert}{E_{\mathrm{down}}}
\newcommand{\hor}{E_{\mathrm{hor}}}

\newcommand{\Tmid}{T_{\mathrm{mid}}}
\newcommand{\Friends}{T_{\mathrm{friends}}}
\newcommand{\ham}{\mathtt{Ham}}

\declaretheorem[numberlike=equation]{Theorem}
\declaretheorem[numberlike=equation]{Lemma}
\declaretheorem[numberlike=equation]{Corollary}
\declaretheoremstyle[bodyfont=\it,qed=$\lozenge$]{defstyle} 

\declaretheorem[numberlike=equation,style=defstyle]{Definition}
\declaretheorem[numberlike=equation]{Claim}
\declaretheorem[numberlike=equation]{Fact}

\hypersetup{
	colorlinks,
	linkcolor=red,
	citecolor=blue,
	urlcolor=green!50!black
}

\title{Distance-preserving Subgraphs of Interval Graphs}

\author{Kshitij Gajjar\thanks{\texttt{kshitij.gajjar@tifr.res.in}}~ and Jaikumar Radhakrishnan\thanks{\texttt{jaikumar@tifr.res.in}}}
\affil{Tata Institute of Fundamental Research, Mumbai}

\begin{document}
	
	\maketitle

	\begin{abstract} 
	We consider the problem of finding small distance-preserving subgraphs of undirected, unweighted interval graphs with $k$ terminal vertices. We prove the following results.
	\begin{enumerate}
	    \item Finding an optimal distance-preserving subgraph is $\NP$-hard for general graphs.
	    \item Every interval graph admits a subgraph with $O(k)$ branching vertices that approximates pairwise terminal distances up to an additive term of $+1$.
	    \item There exists an interval graph $\intgraph$ for which the $+1$ approximation is necessary to obtain the $O(k)$ bound on the number of branching vertices. In particular, any distance-preserving subgraph of $\intgraph$ has $\Omega(k\log k)$ branching vertices.
	    \item Every interval graph admits a distance-preserving subgraph with $O(k\log k)$ branching vertices, implying the $\Omega(k\log k)$ bound is tight for interval graphs.
	    \item There exists an interval graph $\Gzero$ such that every optimal distance-preserving subgraph of $\Gzero$ has $O(k)$ branching vertices and $\Omega(k\log k)$ branching edges, providing a separation between branching vertices and branching edges.
	\end{enumerate}
	The $O(k)$ bound for distance-approximating subgraphs follows from a na\"ive analysis of shortest paths in interval graphs. $\intgraph$ is constructed using bit-reversal permutation matrices. The $O(k\log k)$ bound for distance-preserving subgraphs uses a divide-and-conquer approach. Finally, the separation between branching vertices and branching edges employs Hansel's lemma~\cite{Hansel} for graph covering.
	\end{abstract}
	
	
	
	
	\section{Introduction}
	
	We consider the following problem. Given an undirected, unweighted graph
	$G$ with $k$ vertices designated as terminals, our goal is to
	construct a \emph{small} subgraph $H$ of $G$. Our notion of smallness
	is non-standard: we compare solutions based on the number of vertices
	of degree three or more. We have the following definition.
	
	\begin{Definition}
		Given an undirected, unweighted graph $G=(V,E)$ and a set $R \subseteq
		V$ (the terminals), we say that a subgraph $H(V,E')$ of $G$ is
		distance-preserving for $(G,R)$ if for all terminals $u,v \in R$,
		$d_G(u,v) = d_H(u,v)$, where $d_G$ and $d_H$ denote the distances in
		$G$ and $H$ respectively.  Let $\nbv(H)$ denote the number of vertices
		in $H$ with degree at least three (referred to as \textbf{branching
			vertices}).  Let
		\[ \bv(G,R) = \min_{H} \nbv(H),\]
		where $H$ ranges over all subgraphs of $G$ that are distance-preserving for
		$(G,R)$.  For a family of graphs $\cF$ (such as planar graphs, trees,
		interval graphs), let
		\[
		\bv_{\cF}(k) = \max_{G} \bv(G,R),
		\]
		where $G$ ranges over all graphs in $\cF$, and $R$ ranges over all subsets of $V(G)$ of size $k$.
	\end{Definition}
	In this work, we obtain essentially tight upper and lower bounds on
	$\bv_{\cI}(k)$, where $\cI$ is the class of interval graphs. An interval graph is the intersection graph of a family of intervals on the real line. (See~\autoref{intgraphdef} for a more detailed description.) 
	\begin{Theorem} [Main result] \label{thm:main}
		Let $\cI$ denote the class of interval graphs.
		\begin{enumerate}
			\item[(a)] (Upper bound) $\bv_{\cI}(k) = O(k \log k)$.
			\item[(b)] (Lower bound) There exists a constant $c$ such that for
			each $k$, a positive power of two, there exists an interval graph
			$\intgraph$ with $|R|=k$ terminals such that $\bv(\intgraph,R) \geq c\,
			k \log k$. This implies that $\bv_{\cI}(k) = \Omega(k \log k)$.
		\end{enumerate}
		Parts (a) and (b) imply that $\bv_{\cI}(k) = \Theta(k \log k)$.
	\end{Theorem}
	
	\noindent{\emph{Remark (i).}} Part (a) is constructive. Our proof of the upper bound can be turned into an efficient algorithm that, given an interval graph $G$ on $n$ vertices, produces the required distance-preserving subgraph $H$ of $G$ in running time polynomial in $n$.
	
	\noindent{\emph{Remark (ii).}} Our interval graphs are unweighted. If we consider the family of interval graphs with non-negative weights on their edges ($\cI_w$), then using ~\cite[Section 5]{Ngu}, it is easy to prove that $\bv_{\cI_w}(k) = \Theta(k^4)$ (see~\autoref{corobo} (b)).

	\subsection{Motivation and Related Work}
	
	The problem of constructing small distance-preserving subgraphs bears
	close resemblance to several well-studied problems in graph
	algorithms: graph compression~\cite{Feder}, graph spanners~\cite{Pel,Coppersmith,Bodwin},
	Steiner point removal~\cite{AnupamGupta,Kamma,Filtser}, vertex sparsification~\cite{moses,moitra,englert}, graph homeomorphism~\cite{fortune,lapaugh}, graph contraction~\cite{Daub}, graph sparsification~\cite{spielman,goranci}, etc.
	
	Note that there are several other notions of distance-preserving subgraphs. Our notion of distance-preserving subgraphs is different from that used by Djokovi\'c (and later by Chepoi)~\cite{djok,chepoi}, Nussbaum \emph{et al.}~\cite{nussbaum}, Yan \emph{et al.}~\cite{yan}, or Sadri \emph{et al.}~\cite{sadri}.
	
	For our problem, we emphasize two motivations for studying distance-preserving
	subgraphs, while basing the measure of efficiency on the number of
	branching vertices. First, this problem is closely related to the notion of
	distance-preserving minors introduced by Krauthgamer and
	Zondiner~\cite{KraZon}. Second, although the problem restricted to
	interval graphs is interesting in its own right, it can be seen to
	arise naturally in contexts where intervals represent time periods for
	tasks. Let us now elaborate on our first motivation. Later, we elaborate on the second.
	
	\begin{Definition}
		Let $G(V,E,w)$ be an undirected graph with weight function $w:E
		\rightarrow \mathbb{R}^{\geq 0}$ and a set of terminals $R\subseteq
		V$. Then, $H(V',E',w')$ with $R\subseteq V'\subseteq V$ and weight
		function $w':E' \rightarrow \mathbb{R}^{\geq 0}$ is a
		distance-preserving minor of $G$ if: (i) $H$ is a minor of $G$, and
		(ii) $d_{H}(u,v)=d_{G}(u,v)\, \forall u,v \in R$.
	\end{Definition}
	
	Subsequent work by Krauthgamer, Nguy{\^e}n and Zondiner~\cite{KraZon,Ngu} implies that $\bv_{\cG}(k)=\Theta(k^4)$, where $\cG$ is the family of all undirected graphs (see~\autoref{corobo} (a)).
	
	In this work, we prove that it is $\NP$-hard to determine if $\bv(G,R)\leq m$, when given a general graph $G\in\cG$, a set of terminals	$R\subseteq V(G)$, and a positive integer $m$. A reduction from the set cover problem is described in~\autoref{thm:npcnpc}.
	
	Following the work of Krauthgamer and Zondiner~\cite{KraZon},
	Cheung \emph{et al.}~\cite{Henz} introduced the notion of
	distance-approximating minors.
	\begin{Definition}
		Let $G(V,E,w)$ be an undirected graph with weight function $w:E
		\rightarrow \mathbb{R}^{\geq 0}$ and a set of terminals $R\subseteq
		V$. Then, $H(V',E',w')$ with $R\subseteq V'\subseteq V$ and weight
		function $w':E' \rightarrow \mathbb{R}^{\geq 0}$ is an
		$\alpha$-distance-approximating minor ($\alpha$-$\mathrm{DAM}$) of $G$
		if: (i) $H$ is a minor of $G$, and (ii) $d_{G}(u,v) \le d_{H}(u,v) \le
		\alpha\cdot d_{G}(u,v) \, \forall u,v \in V$.
	\end{Definition}
	In analogy with distance-approximating minors one may ask if interval graphs admit distance-approximating subgraphs with a small number of branching vertices. 
	\begin{Theorem} \label{klintheo}
		Every interval graph $G$ with $k$ terminals admits a subgraph $H$ with $O(k)$ branching vertices such that for all terminals $u$ and $v$ of $G$
		\[ d_G(u,v) \leq d_H(u,v) \leq d_G(u,v) +1 .\]
	\end{Theorem}
	
	We later provide a proof of~\autoref{klintheo} (see the proof of~\autoref{thm:das}).
	
	We now elaborate on our second motivation. The following example\footnote{This is not a real-life problem, though we learnt that minimizing the number of branching vertices in shipping schedules is logistically desirable.} illustrates the relevance of distance-preserving (-approximating) subgraphs for interval graphs.

	\subsection{The Shipping Problem}

	The port of Bandarport is a busy seaport. Apart from ships with routes originating or terminating at Bandarport, there are many ships that dock at Bandarport en route to their final destination. Thus, Bandarport can be considered a hub for many ships from all over the world.
	
	Consider the following shipping problem. A cargo ship starts from some port $X$, and has Bandarport somewhere on its route plan. The ship needs to deliver a freight container to another port $Y$, which is not on its route plan. The container can be dropped off at Bandarport and transferred through a series of ships arriving there until it is finally picked up by a ship that is destined for port $Y$. Thus, the container is transferred from $X$ to $Y$ via some ``intermediate'' ships at Bandarport\footnote{The container cannot be left at the warehouse/storage unit of Bandarport itself beyond a certain limited period of time.}.
	
	However, there is a cost associated with transferring a container from
	one ship to another. This is because each transfer operation requires
	considerable manpower and resources. Thus, the number of ship-to-ship
	transfers that a container undergoes should be as small as possible.
	
	Furthermore, there is an added cost if an intermediate ship receives
	containers from multiple ships, or sends containers to multiple
	ships. This is mainly because of the bookkeeping overhead involved in
	maintaining which container goes to which ship. If a ship is receiving
	all its containers \emph{from} just one ship and sending all those
	containers \emph{to} just one other ship, then the cost associated
	with this transfer is zero (since a container cannot be directed to a
	wrong ship if there is only one option), and this cost increases as
	the number of \emph{to} and \emph{from} ships increases.
	
	Thus, given the docking times of ships at Bandarport, and a small
	subset of these ships that require a transfer of containers between
	each other, our goal is to devise a transfer strategy that meets
	the following objectives.
	
	\begin{itemize}
		\item Minimize the number of transfers for each container. 
		\item Minimize the number of ships that have to deal with multiple transfers.
	\end{itemize}
	
	Representing each ship's visit to the port as an interval on the time
	line, this problem can be modelled using distance-preserving
	(-approximating) subgraphs of interval graphs. In this setting, a shortest path from an earlier interval to a later interval corresponds to a valid sequence of transfers across ships that moves forward in time. The first objective
	corresponds to minimizing pairwise distances between terminals; the
	second objective corresponds to minimizing the number of branching
	vertices.
	
	Let us now quantify this. Suppose that there are a total of $n$ ships
	that dock at the port of Bandarport. Out of these, there are $k$ ships
	that require a transfer of containers between each other (typically
	$k\ll n$). Our results for interval graphs imply the following.
	\begin{enumerate}
		\item If we must make no more than the minimum number of transfers
		required for each container, then there is a transfer strategy in
		which the number of ships that have to deal with multiple transfers
		is $O(k\log k)$.
		
		\item If we are allowed to make \emph{one} more than the minimum
		number of transfers required for each container, then there is a
		transfer strategy in which the number of ships that have to deal
		with multiple transfers is $O(k)$.
		
		\item Neither bound can be improved; that is, there exist scheduling configurations in which $\Omega(k\log k)$ and $\Omega(k)$ ships, respectively, have to deal with multiple transfers.
	\end{enumerate}

	\subsection{Our Techniques}
	
	The linear upper bound mentioned in~\autoref{klintheo} is easy to prove (see~\autoref{thm:das}). However, if we require that distances be preserved exactly, then the problem becomes non-trivial. We now present a broad overview of the techniques involved in proving our main result.
	
	\textit{The Upper Bound:} We may restrict attention to interval graphs that have interval representations where the terminals are intervals of length 0 (their left and right end points are the same) and the non-terminals are intervals of length 1. It is well-known that shortest paths in interval graphs can be constructed using a simple greedy algorithm. We build a subgraph consisting of such shortest paths starting at different terminals and add edges to it so that all inter-terminal shortest paths become available in the subgraph. We use a divide-and-conquer strategy, repeatedly ``cutting'' the graph down the middle into smaller interval graphs. Then we glue the solutions to the two smaller problems together. For this, we need a key observation (which appears to be applicable specifically to interval graphs) that allows one shortest path to ``hop'' onto another. In this, our upper bound method is significantly different from methods used previously for other families of graphs.
	
	\textit{The Lower Bound:} We construct an interval graph and arrange its vertices on a two-dimensional grid instead of the more natural one-dimensional number line. We then show that this grid can be thought of as a matrix, in particular, the bit-reversal permutation matrix (where the ones corresponding to terminals and the zeros to non-terminals). The bit-reversal permutation matrix has seen many applications, most notably in the celebrated Cooley-Tukey algorithm for Fast Fourier Transform~\cite{FFT}. Prior to our work too, it has been used to devise lower bounds (see~\cite{Lynch, Patrascu}). Examining the routes available for shortest paths in our interval graph (constructed using the bit-reversal permutation matrix) requires (i) an analysis of common prefixes of binary sequences, and (ii) building a correspondence between branching vertices and the $k\log k/2$ edges of a $(\log k)$-dimensional Boolean hypercube.

	\section{Distance-preserving Subgraphs of General Graphs}
	
	In this section, we first analyze the problem of finding \emph{optimal} distance-preserving subgraphs of general graphs, and then study distance-preserving subgraphs for \emph{weighted} graphs (including weighted interval graphs).
	
	\subsection{Finding Optimal Distance-preserving Subgraphs}
	
	In this section, we show that the algorithmic task of finding an optimal distance-preserving subgraph of a general graph is $\NP$-hard.
	Consider the following task.
	\begin{description}
		\item{Input:} An undirected, unweighted graph $G$, a set of terminals $R\subseteq V(G)$, and a positive integer $\ell$.
		\item{Output:} Yes, if $(G,R)$ admits a distance-preserving subgraph with at most $\ell$ branching vertices; No, otherwise.
	\end{description}
	
	\begin{figure}
		\begin{center}
				\begin{tikzpicture}[thick,
				every node/.style={},
				fsnode/.style={fill=blue},
				ssnode/.style={fill=green},
				every fit/.style={ellipse,draw,inner sep=-1pt,text width=2cm}
				]
				
				\begin{scope}[start chain=going below,node distance=7mm]
				\foreach \i in {1,2,...,5}
				\node[on chain] (f\i) [] {};
				\end{scope}
				
				\begin{scope}[xshift=4cm,yshift=-0.5cm,start chain=going below,node distance=7mm]
				\foreach \i in {6,7,...,9}
				\node[on chain] (s\i) [] {};
				\end{scope}
				
				\node[] at (0.04,0.4) {$U_1$};
				\node[] at (0.04,-0.6) {$U_2$};
				\node[] at (0,-2.3) {$\vdots$};
				\node[] at (0,-4.3) {$U_{m+1}$};
				
				\vertex at (4,-0.25) [minimum size=3pt, fill, label=left:$S_1$] {};
				\vertex at (4,-0.8) [minimum size=3pt, fill, label=left:$S_2$] {};
				\node[] at (4,-2.1) {$\vdots$};
				\vertex at (4,-3.6) [minimum size=3pt, fill, label=left:$S_m$] {};
				
				\vertex(aa) at (-3.3,-2) [label=left:$t_0$, fill=gray] {};
				\vertex(bb) at (7,-2) [label=right:$t_1$, fill=gray] {};
				
				\draw[thick](aa)--(-0.91,0.23);
				\draw[thick](aa)--(-1.2,-0.58);
				\draw[thick](aa)--(-0.91,-4.24);
				
				\draw[thin](bb)--(4,-0.25);
				\draw[thin](bb)--(4,-0.8);
				\draw[thin](bb)--(4,-3.6);

				\draw(-1.02,0)--(1.02,0);
				\draw(-1.3,-1.1)--(1.3,-1.1);
				\draw(-1.04,-3.9)--(1.04,-3.9);
				
				\node [blue,fit=(f1) (f5),label=above:$\cU$] {};
				\node [green,fit=(s6) (s9),label=above:$\cS$] {};
				
				\end{tikzpicture}
				\caption{The graph $\Gset$ that solves the set cover problem. Each $U_i$ is a copy of $U$. Thus, $\cU$ has $n(m+1)$ vertices and $\cS$ has $m$ vertices. $t_0$ is connected to all vertices of $\cU$, and $t_1$ is connected to all vertices of $\cS$. Vertex $(u,i)$ of $U_i$ has an edge to vertex $S_j$ if and only if $u\in S_j$.}
				\label{fig:setcover}
		\end{center}
	\end{figure}
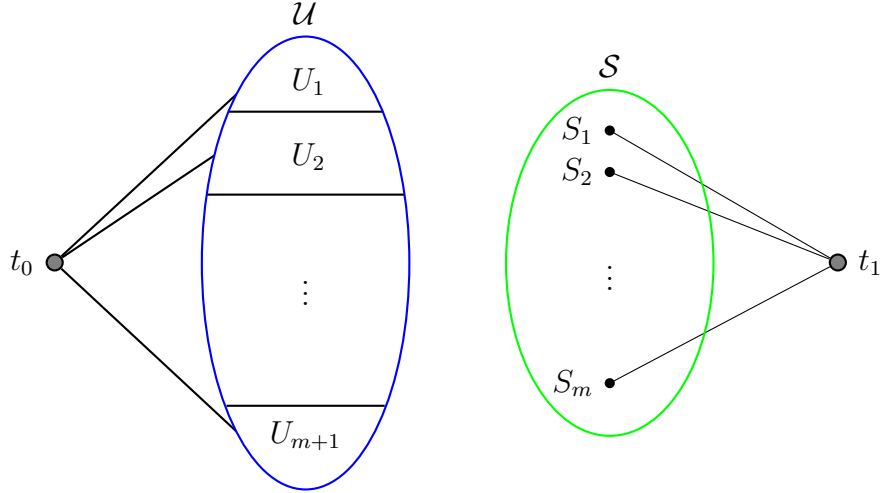

	\begin{Theorem} \label{thm:npcnpc}
		The above decision problem is $\NP$-complete.
	\end{Theorem}
	\begin{proof}
		It is easy to see that the problem is in $\NP$. To show that it is $\NP$-hard, we reduce the set cover problem to the above problem. Consider an instance of the set cover problem on a universe $U$ of size $n$, and subsets
		$S_1,S_2,\ldots, S_m\subseteq U$.
		
		Using this instance of the set cover problem, we construct $\Gset$, a graph on $n(m+1)+m+2$ vertices with $n(m+1)+2$ terminal vertices (\autoref{fig:setcover}). Let $U_1, U_2,\ldots, U_{m+1}$ be $m+1$ copies of $U$.
		\[ U_i = \{ (u,i): u \in U\}.\]
		Let $\cU = \bigcup_i U_i$. Let $\cS = \{S_1,S_2, \ldots, S_m\}$. The vertex set of $\Gset$ is
		$\cU \cup \cS \cup \{t_0,t_1\}$. The edge set of $\Gset$ is $E_0 \cup E_1 \cup
		E_2$, where
		\begin{align*}
		E_0 &= \{ (t_0, (u,i)): (u,i) \in \cU \}; \\
		E_1 &= \{ ((u,i), S_j) : u \in S_j \in \cS \}; \\
		E_2 &= \{ (S_j,t_1): S_j \in \cS \}.
		\end{align*}

		The set of terminals is $\cU \cup \{t_0,t_1\}$.  We claim that $\Gset$ has
		a distance-preserving subgraph with at most $\ell$ non-terminal branching vertices if and only
		if the set cover instance has a cover of size at most $\ell$. The
		\emph{if} direction is straightforward. Simply fix a set cover of size at
		most $\ell$ and consider the subgraph induced by it and the terminals.
		
		For the \emph{only if} direction, suppose there is a distance-preserving
		subgraph $H$ of $\Gset$ that has at most $\ell$ branching vertices.  Clearly,
		in the distance-preserving subgraph $H$, each $(u,i)$ and $t_1$ have a
		common neighbour. If a vertex in $\cS$ has degree at most $2$ in $H$,
		then it can have a neighbour in at most one $U_i$. Since there are
		only $m$ vertices in $\cS$ but $m+1$ sets $U_i$, there is an $i_0$,
		such that each vertex of the form $(u,i_0) \in U_{i_0}$ is a neighbour
		of a branching vertex in $\cS$.  Thus, the (at most $\ell$) branching
		vertices in $\cS$ form a set cover of $U$.
	\end{proof}
	
	\subsection{Distance-preserving Subgraphs of Weighted Graphs}
	
	In this section, we show that $\bv_{\cG_w}(k)=\Theta(k^4)$, where $\cG_w$ is the family of all undirected graphs. This also implies results for unweighted graphs and weighted interval graphs.
	
	\begin{Theorem} \label{thm:bigone}
		If $\cG_w$ is the family of all undirected, weighted graphs, then $\bv_{\cG_w}(k)=\Theta(k^4).$
	\end{Theorem}
	\begin{proof}
		Both the upper bound proof and the lower bound proof for $\bv_{\cG_w}(k)$ follow directly from earlier work of Krauthgamer, Nguy{\^e}n and Zondiner~\cite{Ngu}.
		
		First, we prove that $\bv_{\cG_w}(k)=O(k^4)$. In~\cite[Section 2.1]{Ngu}, they show that every undirected graph on $k$ terminals has a distance-preserving minor with at most $O(k^4)$ vertices. They prove this by pointing out that distance-preserving minors can be constructed by first constructing distance-preserving subgraphs, and then replacing the two edges incident on a vertex of degree two by a single new edge\footnote{Suppose $x$ is a degree two vertex, and $u$ and $v$ are its two neighbours. Then, $(u,x)$ and $(x,v)$ are deleted from the minor, $(u,v)$ is added to the minor (if it does not already exist in the minor), and $w((u,v))\triangleq d(u,v)$.}. The number of vertices in the resulting minor is exactly the number of branching vertices in the distance-preserving subgraph. Thus, $\bv_{\cG_w}(k)= O(k^4)$.
		
		Next, we prove that $\bv_{\cG_w}(k)=\Omega(k^4)$. The weighted planar graph (on $O(k)$ terminal vertices and $\Omega(k^4)$ vertices in total) exhibited in~\cite[Section 5]{Ngu} has only one distance-preserving subgraph, namely the graph itself. Thus, $\bv_{\mathcal{P}_w}(k)=\Omega(k^4)$, where $\mathcal{P}_w$ is the family of all undirected, weighted planar graphs. This implies that $\bv_{\cG_w}(k)=\Omega(k^4)$.
	\end{proof}
	
	\begin{Corollary} \label{corobo}
	(Corollaries of~\autoref{thm:bigone}).
	\begin{enumerate}
		\item[(a)] If $\cG$ is the family of all undirected, unweighted graphs, then $\bv_\cG(k)=\Theta(k^4).$
		\item[(b)] If $\cI_w$ is the family of weighted interval graphs, then $\bv_{\cI_w}(k)=\Theta(k^4).$
	\end{enumerate}
	\end{Corollary}
	\begin{proof}
	    Since $\cG$ and $\cI_w$ are both sub-families of $\cG_w$, the $O(k^4)$ upper bound is straightforward. We now show the lower bound for both the cases.
	    
	    Proof of (a): It is easy to see that the weighted planar graph of~\cite[Section 5]{Ngu} can be made unweighted (by subdividing the edges) so that every distance-preserving subgraph has $\Omega(k^4)$ branching vertices.
	    
		Proof of (b): \autoref{thm:bigone} implies that there exists a weighted graph $G$ such that every distance-preserving subgraph of $G$ has $\Omega(k^4)$ branching vertices. Let $|V(G)|=n$. Add edges of infinte (or very high) weight to $G$ so that the resulting graph is $K_n$, the complete graph on $n$ vertices. Since $K_n$ is an interval graph, this completes the proof.
	\end{proof}

	
	
	
	\section{Interval Graphs}
	
	We work with the following definition of interval graphs.
	\begin{Definition} \label{intgraphdef}
		An interval graph is an undirected graph
		$G(V,E,\mathsf{left},\mathsf{right})$ with vertex set $V$, edge set
		$E$, and real-valued functions $\mathsf{left}:V \rightarrow
		\mathbb{R}$ and $\mathsf{right}:V \rightarrow \mathbb{R}$ such that:
		\begin{itemize}
			\item $\mathsf{left}(x) \leq \mathsf{right}(x) \quad \forall x \in V$;
			\item $(u,v) \in E \Leftrightarrow [\lefty(u),\righty(u)] \cap [\lefty(v),\righty(v)] \neq \emptyset$.
		\end{itemize}
		We order the vertices of the interval graph according to the
		end points of their corresponding intervals. For simplicity, we assume
		that all the end points of the intervals have distinct values. Define
		relations ``$\preceq$'' and ``$\prec$" on the set of vertices $V$ as
		follows.
		\begin{align*}
		u &\preceq v \Leftrightarrow \mathsf{right}(u)\leq \mathsf{right}(v) &&\forall u,v \in V.\\
		u &\prec v \Leftrightarrow \mathsf{right}(u)< \mathsf{right}(v)  &&\forall u,v \in V.
		\end{align*}
		Note that if $u \prec v$, then $u \neq v$.
	\end{Definition}
	
	\subsection{Shortest Paths in Interval Graphs}
	
	In this section, we state some basic properties of shortest paths in interval graphs. It is well-known that one method of constructing shortest paths in interval graphs is the following greedy algorithm. Suppose we need to construct a shortest path from interval $u$ to interval $v$ (assume $u \prec v$). The greedy algorithm starts at $u$. In each step it chooses the next interval that intersects the current interval and reaches farthest to the right. It stops as soon as the current interval intersects $v$. Let $\greedy(u,v)$ be the shortest path produced by this greedy algorithm between $u$ and $v$ ($u \prec v$).

	We now outline some elementary facts about greedy shortest paths, more generally about shortest paths in interval graphs. All of these facts are easy to prove.
	
	\begin{framed}
		{\setlength{\parindent}{0cm}
			\begin{Fact} \label{bridgebasicfact}
				Given an interval graph $G$ and a shortest path (not necessarily a greedy shortest path) $P_G(v_1,v_r)=(v_1,v_2,\ldots,v_r)$ in $G$, if $v_1\prec v_r$, then $v_i\prec v_{i+1}$ for each $1\leq i<r-1$.
			\end{Fact}

			\begin{Fact} \label{bridgepointfact}
				Given an interval graph $G$, a greedy shortest path
				$\greedy(v_1,v_r)=(v_1,v_2,\ldots,v_r)$ in $G$, and a point $a\in
				\mathbb{R}$, let $B_a=\{v_i\in \greedy(v_1,v_r):
				\mathsf{left}(v_i)\leq a\leq\mathsf{right}(v_i)\}$. Then, $|B_a|\leq
				2$.
			\end{Fact}
			
			\begin{Fact} \label{bridgeintervalfact}
				Given an interval graph $G$, and a shortest path (not necessarily a greedy shortest path) $P_G(v_1,v_r)=(v_1,v_2,\ldots,v_r)$ in $G$, and a vertex $x\in
				V(G)$, let $B_x=\{v_i\in P_G(v_1,v_r): (x,v_i)\in E(G)\}$. Then,
				$|B_x|\leq 3$.
			\end{Fact}
			
			\begin{Fact} \label{bridgedomfact}
				Given an interval graph $G$, a greedy shortest path
				$\greedy(v_1,v_r)=(v_1,v_2,\ldots,v_r)$ in $G$, and two vertices
				$x,y\in V(G)$ such that $\mathsf{left}(x)<\mathsf{left}(y)$ and
				$\mathsf{right}(x)>\mathsf{right}(y)$. Then, $y\in \greedy(v_1,v_r)$
				if and only if $v_1=y$ or $v_r=y$.
			\end{Fact}

	}\end{framed}
	We now proceed to prove the $O(k)$ upper bound for distance-approximating subgraphs of interval graphs.
	
	\subsection{Distance-approximating Subgraphs of Interval Graphs}
	
	In this section, we show that a simple greedy technique yields a distance-approximating subgraph for any interval graph. Let us restate~\autoref{klintheo}.
	
	\begin{Theorem} \label{thm:das}
		If $\mathcal{I}$ is the family of all interval graphs, then there exists a subgraph $H$ of $G$ such that $\nbv(H)=O(k)$ and for all terminals $u$ and $v$, we have $d_G(u,v)\leq d_H(u,v)\leq 1+d_G(u,v)$.
	\end{Theorem}
	
	Let $G(V,E,\mathsf{left},\mathsf{right})$ be an interval graph on $k$ terminals indexed by the set $[k]$. For any two vertices $u\preceq v$ of $G$, let $\greedy(u,v)$ be the greedy shortest path between $u$ and $v$, as defined in the previous section. 
	For each $1\leq i<k$, define the tree $T_i$ as follows.
	
	$$T_i=\displaystyle\bigcup_{i<j\leq k}\greedy(i,j)$$
	
	
	Thus, $T_1$ is a shortest-path tree rooted at terminal $1$. 
	We are now set to define $H_1$. This is the distance-approximating subgraph of $G$.
	
	$$H_1=T_1\cup\{(v,i)\in E(G) : 1<i<k, v\in V(T_1)\}$$
	
	Assume that $v_{\mathrm{last}}=k$. Then, $H_1$ may alternatively be defined as follows.
	
	$$H_1=\greedy(1,k)\cup\{(v,i)\in E(G) : 1<i<k, v\in V(T_1)\}$$
	
	It is easy to check that both these definitions are equivalent. The following theorem proves that $H_1$ approximates terminal distances in $G$ up to an additive term of $+1$.
	
	\begin{Lemma} \label{lemma+1approx}
		$d_G(i,j)\leq d_{H_1}(i,j)\leq d_G(i,j)+1 \qquad \forall\, 1\leq i<j\leq k$.
	\end{Lemma}
	\begin{proof}
		For $i=1$, $d_G(1,j)=d_{H_1}(1,j)$ since $T_1\subseteq H_1$. Also when $(i,j)\in E(G)$, it is easily verifiable that $d_{H_1}(i,j)\leq 2$.
		
		Now suppose $i\neq 1$. We show that for any $j$ such that $i<j\leq k$ and $(i,j)\notin E(G)$, $d_{H_1}(i,j)\leq d_G(i,j)+1$. Let $\greedyH(i,k)$ be the greedy shortest path from $i$ to $k$ in $H_1$. For integer $p\geq 1$, let  $v_G(i,p)$ be the $p$-th vertex on the path $\greedy(i,k)$ ($i$ itself being the 0-th vertex). $v_{H_1}(i,p)$ is similarly defined. Note that $\mathsf{right}(v_G(i,p))\geq\mathsf{right}(v_{H_1}(i,p))$, with equality occurring when $v_G(i,p)=v_{H_1}(i,p)$.
		
		Suppose $d_G(i,j)=p$. Then, $(v_G(i,p-1),j)\in E(G)$. Using~\autoref{claim+1approx}, we know that either $v_G(i,p-1)=v_{H_1}(i,p-1)$ or $(v_G(i,p-1),v_{H_1}(i,p-1))\in E(G)$. In the first case, $d_G(i,j)=d_{H_1}(i,j)=p$ and we are done. In the second case, there is a path of length at most $2$ from $v_{H_1}(i,p-1)$ to $j$. Thus, $d_{H_1}(i,j)\leq (p-1)+2=p+1$. This completes the proof.
	\end{proof}
	
	\begin{Claim} \label{claim+1approx}
		Let $v_G(i,p)$ and $v_{H_1}(i,p)$ be as defined in the proof of~\autoref{lemma+1approx}. Then for all $p\geq 1$, either $v_G(i,p)=v_{H_1}(i,p)$ or $(v_G(i,p),v_{H_1}(i,p))\in E(G)$. 
	\end{Claim}
	
	\begin{proof}
		We prove this claim by inducting on $p$. For $p=1$, the claim is trivially true. Our goal is to prove that the claim is true for $p=r+1$, assuming that the claim is true for $p=r$. Thus, our induction hypothesis is that either $v_G(i,r)=v_{H_1}(i,r)$ or $(v_G(i,r),v_{H_1}(i,r))\in E(G)$. In the first case, we have $v_G(i,r+1)=v_{H_1}(i,r+1)$, and we are done. In the second case, assume that $(v_G(i,r),v_{H_1}(i,r))\in E(G)$. Then, we have the following.
		\begin{align*}
		v_G(i,r+1)&=\underset{x}{\operatorname{argmax}}\{\mathsf{right}(x)\mid x\in V(G),(v_G(i,r),x)\in E(G)\}\\
		v_{H_1}(i,r+1)&=\underset{x}{\operatorname{argmax}}\{\mathsf{right}(x)\mid x\in V(H_1),(v_{H_1}(i,r),x)\in E(H_1)\}
		\end{align*}
		If $v_{G}(i,r+1)=v_{H_1}(i,r+1)$, then we are done. Otherwise, $\mathsf{left}(v_{H_1}(i,r+1))<\mathsf{right}(v_{H_1}(i,r))<\mathsf{right}(v_{G}(i,r))\leq\mathsf{right}(v_{H_1}(i,r+1))$. Thus, the point $\mathsf{right}(v_{G}(i,r))$ is present in the interval corresponding to $v_{H_1}(i,r+1)$ as well as in the interval corresponding to $v_{G}(i,r+1)$, which implies that $(v_G(i,r+1),v_{H_1}(i,r+1))\in E(G)$. This completes the proof of the claim.
	\end{proof}
	
	Thus, $H_1$ approximates terminal distances in $G$ up to an additive term of $+1$. We now prove that the number of branching vertices in $H_1$ is linear in $k$.
	\begin{Lemma} \label{lemma+1ub}
		$H_1$ has $O(k)$ branching vertices.
	\end{Lemma}
	\begin{proof}
		For $i\in[k]$, let $B(1)_i=\{v\in \greedy(1,k)\mid (i,v)\in E(G)\}$. Then by~\autoref{bridgeintervalfact}, $|B(1)_i|\leq 3$. In other words, each terminal can contribute at most $3$ branching vertices to $H_1$. Summing over all terminals,
		
		$$\sum\limits_{i=1}^{k}{|B(1)_i|}\leq 3k$$
		
		$\greedy(1,k)$ is a simple path and thus contributes no branching vertices of its own to $H_1$. Since $H_1=\greedy(1,k)\cup\{(v,i)\in E(G) : 1<i<k, v\in V(T_1)\}$, $H_1$ has at most $3k$ branching vertices, completing the proof.
	\end{proof}
	
	\autoref{lemma+1approx} and~\autoref{lemma+1ub} together complete the proof of~\autoref{thm:das}. Finally, we prove that this upper bound is tight by providing a matching lower bound.
	\begin{Lemma} \label{lemma+1lb}
		For every positive integer $k$, there exists an interval graph $\Ghard$ on $k$ terminals such that if $H_1$ is a subgraph of $\Ghard$ and $H_1$ approximates distances in $\Ghard$ up to an additive distortion of $+1$, then $H_1$ has $\Omega(k)$ branching vertices.
	\end{Lemma}
	\begin{proof}
		Let us describe the construction of $\Ghard(V,E,\mathsf{left},\mathsf{right})$. Fix $\epsilon=0.01$. $\Ghard$ has $2k-2$ non-terminal vertices $\{v_1,v_2,\ldots,v_{2k-2}\}$ and $k$ terminals vertices indexed by the set $[k]$.
		\begin{align*}
		&\mathsf{left}(v_i)=i-\epsilon, &&\mathsf{right}(v_i)=i+\epsilon+1 &\forall\, 1\leq i\leq 2k-2.\\
		&\mathsf{left}(j)=2j-1.5, &&\mathsf{right}(j)=2j-0.5 &\forall\, j\in[k].
		\end{align*}
		Suppose $H_1$ approximates distances in $\Ghard$ up to an additive distortion of $+1$. For odd $i$, $(v_i,v_{i+1})\in V(H_1)$ (otherwise the terminals become disconnected in $H_1$). For even $i$, define the set $S$ as follows (let $j=i/2$).
		$$S=\{(v_{2j},v_{2j+1}) : 1\leq j\leq k-2, (v_{2j},v_{2j+1})\in E(H_1)\}$$
		Thus, $|S|\leq k-2$. Consider any $(v_{2j},v_{2j+1})\in S$. Then either $(v_{2j},j+1)\in E(H_1)$ or $(j+1,v_{2j+1})\in E(H_1)$ (otherwise $j+1$ becomes isolated in $H_1$). Since $(v_{2j-1},v_{2j})\in E(H_1)$ and $(v_{2j+1},v_{2j+2})\in E(H_1)$, either $v_{2j}$ or $v_{2j+1}$ must be a branching vertex in $H_1$. Thus, for every edge in $S$, at least one of its end points must be a branching vertex. Using~\autoref{claimsizeofS}, $|S|\geq k-3$. Combined with the fact that all the edges of $S$ are vertex disjoint, this means that $H_1$ has at least $k-3$ branching vertices, completing the proof.
	\end{proof}
	
	\begin{Claim} \label{claimsizeofS}
		Let $S$ be as defined in the proof of~\autoref{lemma+1lb}. Then $|S|\geq k-3$.
	\end{Claim}
	\begin{proof}
		Suppose $|S|<k-3$. Then, there exist terminals $j_1,j_2$ such that $1\leq j_1<j_2\leq k-2$ and $(v_{2j_1},v_{2j_1+1})\notin E(H_1), (v_{2j_2},v_{2j_2+1})\notin E(H_1)$. This means that $d_{H_1}(v_{2j_1},v_{2j_1+1})=2=1+d_{\Ghard}(v_{2j_1},v_{2j_1+1})$ and $d_{H_1}(v_{2j_2},v_{2j_2+1})=2=1+d_{\Ghard}(v_{2j_2},v_{2j_2+1})$. Thus, $d_{H_1}(1,k)\geq 2+d_{\Ghard}(1,k)$, and $H_1$ does not approximate distances in $\Ghard$ up to an additive distortion of $+1$, which is a contradiction. This completes the proof of the claim.
	\end{proof}
	Note that our proof of the $O(k)$ upper bound naturally translates into an algorithm. In other words, given an interval graph $G$ on $n$ vertices, it produces a distance-approximating subgraph $H$ of $G$ in running time polynomial in $n$. We now move on to distance-preserving subgraphs of interval graphs.
	
	
	
	
	\section{Proof of the Upper Bound}
	
	In this section, we show that any interval graph $G$ with $k$ terminals has a distance-preserving subgraph with $O(k\log k)$ branching vertices, which is simply~\autoref{thm:main} (a), restated here for completeness.
	\begin{Theorem} \label{thm:ub}
		If $\cI$ is the family of all interval graphs, then $\bv_\cI(k)=O(k\log k).$
	\end{Theorem}
	
	The following notation will be used in our proof. Given real numbers $a,b \in \mathbb{R}$ such that $a\leq b$, let $G[a,b]$ be the induced subgraph on those vertices $v$ of $G$ such that $[\lefty(v),\righty(v)] \cap [a,b] \neq \emptyset$.  Similarly, let $G[a,b)$ be the induced subgraph on those vertices $v$ of $G$ such that $[\lefty(v),\righty(v)] \cap [a,b) \neq \emptyset$.
	
	We first prove the upper bound for a special case of interval graphs, and later show that the same upper bound holds (up to constants) for all interval graphs.
	
	\subsection{Unit Interval Graphs with Point Terminals}
	
	Let $G$ be an interval graph on $k$ terminals such that all terminals in $G$ are zero-length intervals (or point intervals) and all non-terminals are unit intervals. Our goal is to obtain a distance-preserving subgraph $H$ of $G$ with $O(k \log k)$ branching vertices. Note that the $H$ that we obtain is not necessarily an interval graph. This is because $H$ need not be an induced subgraph of $G$.
	
	Consider the greedy path $\greedy(t_i,t_k)$ ($i < k$), where $t_k$ is
	the rightmost terminal. Our distance-preserving subgraph includes
	greedy paths from $t_i$ to $t_k$ for all $1\leq i<k$.  Let
	\begin{equation}
	H_0 = \bigcup_{1\leq i<k} \greedy(i,k).
	\end{equation}
	Now, $H_0$ already provides for shortest paths from each terminal
	$t_i$ to $t_k$. In fact, it can be viewed as a shortest path tree with
	root $t_k$, but constructed backwards.  Thus, the total number of
	branching vertices in $H_0$ is $O(k)$.  We still need to arrange for
	shortest paths between other pairs of terminals $(t_i,t_j)$.  The path
	$\greedy(t_i,t_j)$ (for $i< j < k$) is either entirely contained in $\greedy(t_i,t_k)$, or it follows $\greedy(t_i,t_k)$ until it reaches a neighbour of $t_j$ and then \emph{branches off} to connect to $t_j$. We
	can consider including all paths of the form $\greedy(t_i,t_j)$ in
	$H_0$. That is, we need to link each such $t_j$ to vertices from
	$H_0$ so that each path $\greedy(t_i,t_j)$ becomes available. If this
	is done without additional care, we might end up introducing
	$\Omega(k)$ additional branching vertices per terminal, and
	$\Omega(k^2)$ branching vertices in all, far more than we claimed.
	
	The crucial idea for overcoming this difficulty is contained in the
	following lemma.
	\begin{Lemma} \label{lm:link}
		Suppose $v \prec w$ and $d(v,w)=1$. Let $(v,v_1,v_2,\ldots,v_\ell)$
		and $(w,w_1,w_2,\ldots,w_{\ell'})$ be greedy shortest paths starting
		from $v$ and $w$ respectively. Suppose $\righty(v_{\ell}) <
		\righty(w_{\ell'})$. Then, $\ell \leq \ell'$.
	\end{Lemma}
	\begin{proof}
		Since $d(v,w)=1$, the greedy strategy reaches at least as far in $j+1$
		steps from $v$ as it does in $j$ steps from $w$. Suppose for
		contradiction that $\ell > \ell'$ (that is, $\ell \geq \ell'+1$).
		Then, we have $\righty(w_{\ell'}) \leq \righty(v_{\ell'+1}) \leq
		\righty(v_{\ell})$, contradicting our assumption that
		$\righty(v_{\ell}) < \righty(w_{\ell'})$.
	\end{proof}
	The above lemma is crucial for the construction of our subgraph
	$H$. 
	For example, suppose $t_i$ and $t_j$ both need to reach $t_r$ via a shortest path. Suppose $(w_i,t_r)$ is the last edge of
	$\greedy(t_i,t_r)$ and $(w_j,t_r)$ is the last edge of $\greedy(t_j,t_r)$. We claim that it is sufficient to include \textbf{only} one of these edges in $H$. If $\righty(w_j) < \righty(w_i)$, then it is enough to include the edge $(w_j,t_r)$ in $H$; as long as $t_i$ has a shortest path to $w_j$, this edge serves for shortest paths \emph{to} $t_r$ \emph{from} both $t_i$ and $t_j$.  In the construction below, we add links to the greedy paths of $H_0$ so that we need to provide only one such edge per terminal. This idea forms the basis of the divide-and-conquer strategy which we present below.
	
	Suppose $G$ has $2\ell$ terminals. We find a point $x$ so that both
	$\Gleft = G[-\infty,x]$ and $\Gright = G[x, \infty]$ have $\ell$
	terminals. By induction, we find distance-preserving subgraphs
	$\Hleft$ and $\Hright$ of $\Gleft$ and $\Gright$ with at most $f(\ell)$
	branching vertices each. The union of $\Hleft$ and $\Hright$ has just
	$2f(\ell)$ branching vertices, but it does not yet guarantee shortest
	paths from terminals in $\Hleft$ to terminals in $\Hright$. Using~\autoref{lm:link} and the discussion above, we connect each terminal $t_j$ in $\Hright$ to \textbf{only} one of the greedy shortest paths
	of terminals from $\Hleft$, and ensure that shortest paths to $t_j$ are
	preserved from all terminals $t_i$ in $\Hleft$. This creates
	$O(\ell)$ additional branching vertices and give us a recurrence of the
	form
	\[ f(2\ell) \leq 2 f(\ell) + O(\ell),\]
	and the desired upper bound of $O(k \log k)$.  Unfortunately, there
	are technical difficulties in implementing the above strategy as
	stated. It is therefore helpful to augment $H_0$ by adding all
	greedy paths $\greedy(t_i,t_j)$, where $d(i,j) \leq 4$. As a result,
	for each terminal $t_i$, the first three vertices on
	$\greedy(t_i,t_k)$ might become branching vertices. In all, this adds
	a \emph{one-time cost} of $O(k)$ branching vertices to our
	subgraph. We now present the argument formally.

	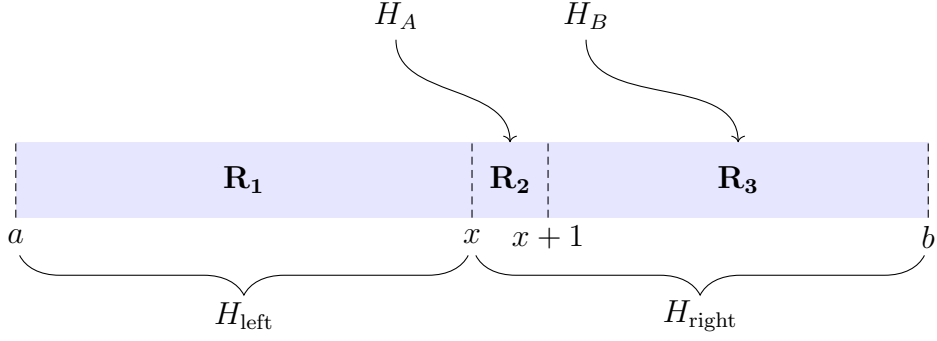
\begin{figure}
		\begin{center}
			\begin{tikzpicture}

			\fill[fill=blue!10](0,0) rectangle +(12,1);
			
			\draw[densely dashed](0,0) -- (0,1);
			\draw[densely dashed](6,0) -- (6,1);
			\draw[densely dashed](7,0) -- (7,1);
			\draw[densely dashed](12,0) -- (12,1);
			
			\node at (0,-0.25) {$a$};
			\node at (6,-0.25) {$x$};
			\node at (7,-0.25) {$x+1$};
			\node at (12,-0.25) {$b$};
			
			\node at (3,0.5) {$\bf R_1$};
			\node at (6.5,0.5) {$\bf R_2$};
			\node at (9.5,0.5) {$\bf R_3$};
			
			\draw [decorate,decoration={brace,amplitude=15pt}] (5.95,-0.5) -- (0.05,-0.5) node [black,midway,yshift=-0.8cm] {$\Hleft$};
			\draw [decorate,decoration={brace,amplitude=15pt}] (11.95,-0.5) -- (6.05,-0.5) node [black,midway,yshift=-0.8cm] {$\Hright$};
			
			\node[anchor=north] at (5,3) (ha) {$H_A$};
			\draw (ha) edge[out=-90,in=90,->] (6.5,1);
			\node[anchor=north] at (7.5,3) (hb) {$H_B$};
			\draw (hb) edge[out=-90,in=90,->] (9.5,1);

			\end{tikzpicture}
			\caption{The interval graph $G[a,b]$ has $2\ell$ terminals and is ``cut" into three regions, $R_1, R_2$ and $R_3$. By induction, $\Hleft$ preserves distances when both terminals lie in $R_1$, and $\Hright$ preserves distances when both terminals lie in $R_2\cup R_3$. In addition, $H_A$ preserves distances when one terminal lies in $R_1$ and the other in $R_2$ (by introducing at most $O(\ell)$ additional branching vertices in $R_2$), and $H_B$ preserves distances when one terminal lies in $R_1$ and the other in $R_3$ (by introducing at most $O(\ell)$ additional branching vertices in $R_3$).}
			\label{fig:oned}
		\end{center}
	\end{figure}

	For each $(a,b)$, let $f(a,b)$ be the minimum number of non-terminals in a
	subgraph $H^*$ of $G[a,b]$ such that $H_0 \cup H^*$ preserves all
	inter-terminal distances in $G[a,b]$; let
	\[ f(\ell) = \max_{(a,b)} f(a,b),\]
	where $(a,b)$ ranges over all pairs such that $G[a,b]$ has at most
	$\ell$ terminals. The following lemma is the basis of our induction.
	\begin{Lemma}\label{lm:induction}
		(i) $f(1) = 0$; (ii) $f(2\ell) \leq 2 f(\ell) + O(\ell)$.
	\end{Lemma}
	\begin{proof}
		Part (i) is trivial.  For part (ii), fix a pair $(a,b)$ such that
		$G[a,b]$ has at most $2\ell$ terminals.  If $b-a \leq 1$, $H_0$
		already preserves distances between every two terminals in
		$G[a,b]$. So, we may take $H^*$ to be empty. Now assume that $b-a>
		1$.  Pick $x \in [a,b]$ as large as possible such that (i) $b-x\geq
		1$, and (ii) $G[x,b]$ has at least $\ell$ terminals.
		
		Let $\Gleft = G[a,x)$ and $\Gright = G[x,b]$.  Since $\Gright$ has at
		least $\ell$ terminals, $\Gleft$ has at most $\ell$ terminals. So,
		we obtain (by induction) a subgraph $\Hleft$ of $G[a,b]$ with at
		most $f(\ell)$ non-terminals, such that $H_0 \cup \Hleft$ preserves all
		inter-terminal distances in $\Gleft$. If $b-x>1$, then $\Gright$ has
		exactly $\ell$ terminals, and we obtain by induction a subgraph
		$\Hright$ of $G[a,b]$ with at most $f(\ell)$ non-terminals such that $H_0 \cup
		\Hright$ preserves all inter-terminal distances in $G[x,b]$. If $b-x
		=1$, then we may take $\Hright$ to be empty (for $H_0$ already
		preserves inter-terminal distances in $G[x,b]$).
		
		Our final subgraph $H^*$ shall be of the form $\Hleft \cup \Hright \cup H_A \cup H_B$, where $H_A$ and $H_B$ are defined as follows. (Refer to~\autoref{fig:oned}.) Let us first define $H_A$. Let $\Pleft$ be the set of greedy paths from the terminals in $\Hleft$ to the terminal $t_k$.  Let $V_A$ be the set of all non-terminal intervals of $\Pleft$ that intersect with the interval $[x,x+1]$.  It is easy to see that any path in $\Pleft$ contributes at most $4$ non-terminals to $V_A$. So, $|V_A| \leq 4\ell$. Let $H_A$ be the subgraph of $G[a,b]$ induced by $V_A$ and the terminals in $G[x,x+1]$.
		
		Note that $H_0 \cup \Hleft \cup \Hright \cup H_A$ preserves all inter-terminal distances in $G[a,x+1]$ as well as all inter-terminal distances in $G[x+1,b]$. In fact, it does more. For each terminal $t_i$ in $G[a,x)$, let $v_i$ be the last vertex on the greedy path $\greedy(t_i,t_k)$ that is in $V_A$. Then, the above graph contains the greedy shortest path from every terminal $t_j$ in $G[a,x]$ to $v_i$.
		
		Now, it only remains to ensure that distances between terminals in $G[a,x)$ and terminals in $G[x+1,b]$ are preserved. Let us now define $H_B$. For each terminal $t_j$ in $G[x+1,b]$, let $v$ be the earliest interval (with respect to $\prec$) of $\Pleft$ that contains $t_j$. Then, we include the edge $(v,t_j)$ in $H_B$. Thus, $H_B$ contains at most one non-terminal per vertex in $G[x+1,b]$; that is, at most $2\ell$ non-terminals in all. This completes the description of $H_A$ and $H_B$. The final subgraph is $H^*=\Hleft \cup \Hright \cup H_A \cup H_B$.
		
		\begin{Claim} \label{joinnconn}
			Let $t_i$ be a terminal in $G[a,x)$ and $t_r$ be a terminal in $G[x,b]$. Then,
			$H = H_0 \cup H^*$ preserves the distance between terminal $t_i$ and $t_r$.
		\end{Claim}
		\noindent \emph{Proof of~\autoref{joinnconn}.} Let $v$ be the vertex that we
		attached to $t_r$ in $H_B$. If $v$ is on $\greedy(t_i,t_k)$, then it
		follows that $\greedy(t_i,t_r)$ is in $H$, and we are done. So we assume that $v$ is not on $\greedy(t_i,t_k)$. Then,
		let $j\neq i$ be such that $v \in \greedy(t_j,t_k)$. Then, we have
		paths
		\begin{align*}
		P_G(t_i,t_r)&=(t_i,w_1,w_2,\ldots,w_{p},w_{p+1},\ldots,w_{\ell'},t_r);\\
		P_H(t_i,t_r)&=(t_i,w_1,w_2,\ldots,w_{p},v_{q+1},\ldots,v_{\ell}=v,t_r),
		\end{align*}
		where $v_{q+1}$ is the last vertex on $\greedy(t_j,t_k)$ in $G[x,x+1]$, and $w_{p}$ is the first vertex on $\greedy(t_i,t_r)$  such that $(w_{p},v_{q+1})\in E(G)$. From the construction of $H_A$, $(w_{p},v_{q+1})\in E(H)$. Following $v_q$, $(v_{q+1},\ldots, v_{\ell}=v,t_r)$ are the subsequent vertices on
		$\greedy(t_j,t_r)$. Note that: (i) $v_{q+1} \prec w_{p+1}$ (otherwise $v$ is on $\greedy(t_i,t_k)$), (ii) $d(v_{q+1},w_{p+1})=1$ (both intervals contain $\righty(w_{p})$), and (iii) $\righty(v_{\ell}) < \righty(w_{\ell'})$ (since $v$ is the earliest interval of $\Pleft$ that contains $t_j$). By~\autoref{lm:link}, $\ell-q-1\leq\ell'-p-1$. Thus, $P_H(t_i,t_r)$ is no longer than $\greedy(t_i,t_r)$.
	\end{proof}
	
	We can now complete the proof of the upper bound. By~\autoref{lm:induction}, there is a subgraph $H^*$ of $G$ such that $H=H_0 \cup H^*$ preserves all inter-terminal distances in $G$, $H_0$ has $O(k)$ branching vertices and $H^*$ has $O(k\log k)$ non-terminals. It follows that $H$ has $O(k \log k)$ branching vertices.
	
	\subsection{Generalizing to all Interval Graphs}
	
	In this section, we prove the upper bound for general interval graphs. In particular, we show that any interval graph can be reduced to the special case of the previous section. Given an interval graph $G$ on $k$ terminals, we produce a slightly modified interval graph $G'$ on $2k$ terminals, such that $G'$ has point terminals and unit non-terminals.

	\begin{enumerate}
		\item Initially, $G'=G$.
		\item For each terminal $t\in V(G)$, add two vertices $t_{\mathrm{left}}$ and $t_{\mathrm{right}}$ to $V(G')$ such that
		\[
		\lefty(t_{\mathrm{left}})=\righty(t_{\mathrm{left}})=\lefty(t),\\
		\lefty(t_{\mathrm{right}})=\righty(t_{\mathrm{right}})=\righty(t).
		\]
		Thus, $t_{\mathrm{left}}$ and $t_{\mathrm{right}}$ are point intervals, or intervals of length zero. \item Designate $t_{\mathrm{left}}$ and $t_{\mathrm{right}}$ as terminals, and $t$ as a non-terminal. Thus, $G'$ now has $2k$ terminals.
		\item For each non-terminal $x\in V(G')$, if there exists $y\in V(G')$ such that $\mathsf{left}(y)<\mathsf{left}(x)$ and $\mathsf{right}(y)>\mathsf{right}(x)$, then delete $x$ from $G'$.
	\end{enumerate}
	
	Thus, after this pre-processing, the non-terminals of $G'$ have the following property: for any pair of non-terminals $u$ and $v$, $\lefty(u)\leq\lefty(v)\Leftrightarrow\righty(u)\leq\righty(v)$. Gardi~\cite{Gardi} shows that an interval graph that possesses this property can be equivalently represented by a unit interval graph. Thus, the induced subgraph on the non-terminals of $G'$ has a unit interval representation. Using Helly's theorem~\cite[Page 102]{helly} stated below, it is easy to see that the (zero-length) terminals can also be placed in this unit interval representation such that it represents $G'$.
	
	\begin{Theorem} [Helly~\cite{helly}, Page 102]
		Given a finite set of intervals on the real line such that each pair of intervals has a nonempty intersection, there exists a point that lies on all the intervals.
	\end{Theorem}

	The final $G'$ is an interval graph (not necessarily a unit interval graph) such that all terminal intervals of $G'$ have length $0$, and all non-terminal intervals of $G'$ have length $1$. Thus, $G'$ now possesses the structure required for the construction of a distance-preserving subgraph $H'$ with $O(k\log k)$ branching vertices, as described in the proof of~\autoref{thm:ub} in the previous section.
	
	We now construct $H$ from $H'$, such that $H$ is a distance-preserving subgraph of $G$.
	
	\begin{enumerate}
		\item Initially, $H=H'$.
		\item For each terminal $t\in V(G)$, if $t\notin V(H)$, then add $t$ to $V(H)$.
		\item For each terminal $t\in V(H)$, modify the neighbourhood of $t$ in $H$ as follows.
		\[ \nbhd_H(t)=\nbhd_H(t_{\mathrm{left}})\cup\nbhd_H(t_{\mathrm{right}})\cup\nbhd_H(t). \]
		\item For each terminal $t\in V(H)$, delete $t_{\mathrm{left}}$ and $t_{\mathrm{right}}$ from $H$.
		\item For terminal pairs $t_1$ and $t_2$ such that $d_G(t_1,t_2)=1$, if $(t_1,t_2)\notin E(H)$, then add $(t_1,t_2)$ to $E(H)$.
	\end{enumerate}
	
	It is easy to see that the number of branching vertices in $H$ at most $O(k)$ more than the number of branching vertices in $H'$. Also, for terminal pairs $u$ and $v$ such that $d_G(u,v)>1$ (assume $u\prec v$), a shortest path from $u_{\mathrm{right}}$ to $v_{\mathrm{left}}$ is a shortest path from $u$ to $v$. Thus, $H$ is a distance-preserving subgraph of $G$ with $O(k\log k)$ branching vertices, completing the proof of~\autoref{thm:ub}.

	
	

	\section{Proof of the Lower Bound}
	
	In this section, we show that there exists an interval graph $\intgraph$ such that any distance-preserving subgraph of $\intgraph$ has $\Omega(k\log k)$ branching vertices, which is simply~\autoref{thm:main} (b), restated here for completeness.
	\begin{Theorem} \label{thm:lb}
		If $\cI$ is the family of all interval graphs, then $\bv_\cI(k)=\Omega(k\log k).$
	\end{Theorem}
	
	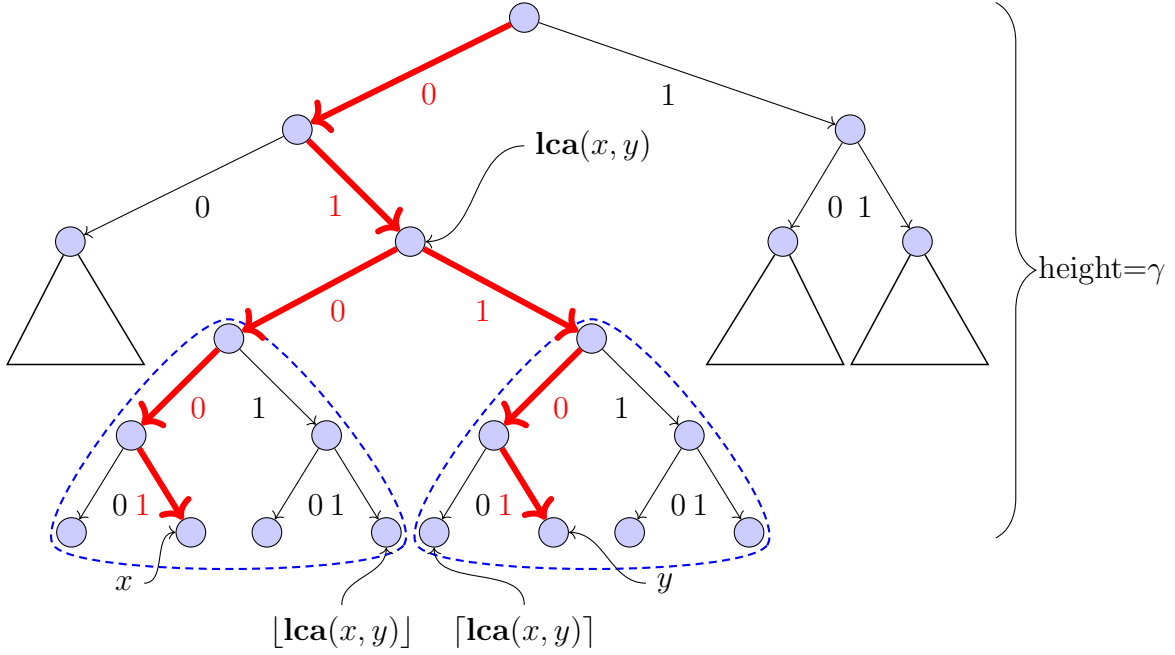
\begin{figure}
		\begin{center}
			\tikzstyle{vis}=[circle,draw=black,fill=blue!20]
			\begin{tikzpicture}
			
			\draw [densely dashed, thick, blue] plot [smooth cycle] coordinates {(-3.89,-4) (-1.59,-7) (-6.19,-7)};
			\draw [densely dashed, thick, blue] plot [smooth cycle] coordinates {(0.89,-4) (3.19,-7) (-1.41,-7)};
			
			\node[vis] (v)   [] {};
			
			\node[vis] (v0)  [below left=1.2 and 2.7 of v] {};
			\node[vis] (v1)  [below right=1.2 and 4 of v] {};
			
			\node[vis] (v00) [below left=1.2 and 2.7 of v0] {};
			\node[vis] (v01) [below right=1.2 and 1.2 of v0] {};
			\node[vis] (v10) [below left=1.2 and 0.6 of v1] {};
			\node[vis] (v11) [below right=1.2 and 0.6 of v1] {};
			
			\node[vis] (v010) [below left=1 and 2.1 of v01] {};
			\node[vis] (v011) [below right=1 and 2.1 of v01] {};
			
			\node[vis] (v0100)  [below left=1 and 1 of v010] {};
			\node[vis] (v0101)  [below right=1 and 1 of v010] {};
			
			\node[vis] (v01000) [below left=1 and 0.5 of v0100] {};
			\node[vis] (v01001) [below right=1 and 0.5 of v0100] {};
			\node[vis] (v01010) [below left=1 and 0.5 of v0101] {};
			\node[vis] (v01011) [below right=1 and 0.5 of v0101] {};
			
			\node[vis] (v0110)  [below left=1 and 1 of v011] {};
			\node[vis] (v0111)  [below right=1 and 1 of v011] {};
			
			\node[vis] (v01100) [below left=1 and 0.5 of v0110] {};
			\node[vis] (v01101) [below right=1 and 0.5 of v0110] {};
			\node[vis] (v01110) [below left=1 and 0.5 of v0111] {};
			\node[vis] (v01111) [below right=1 and 0.5 of v0111] {};
			
			\draw [->,red, line width=0.8mm] (v)   -- node[below right=-0.4mm and -0.4mm] {0} (v0);
			\draw [->] (v)   -- node[below left]  {1} (v1);
			
			\draw [->] (v0)  -- node[below right] {0} (v00);
			\draw [->,red, line width=0.8mm] (v0)  -- node[below left=-0.4mm and -0.4mm]  {1} (v01);
			\draw [->] (v1)  -- node[below right] {0} (v10);
			\draw [->] (v1)  -- node[below left]  {1} (v11);
			
			\draw [->,red, line width=0.8mm] (v01) -- node[below right=-0.4mm and -0.4mm] {0} (v010);
			\draw [->,red, line width=0.8mm] (v01) -- node[below left=-0.4mm and -0.4mm]  {1} (v011);
			
			\draw [->,red, line width=0.8mm] (v010)  -- node[below right=-0.4mm and -0.4mm] {0} (v0100);
			\draw [->] (v010)  -- node[below left]  {1} (v0101);
			\draw [->,red, line width=0.8mm] (v011)  -- node[below right=-0.4mm and -0.4mm] {0} (v0110);
			\draw [->] (v011)  -- node[below left]  {1} (v0111);
			\draw [->] (v0100) -- node[below right] {0} (v01000);
			\draw [->,red, line width=0.8mm] (v0100) -- node[below left=-0.4mm and -0.4mm]  {1} (v01001);
			\draw [->] (v0101) -- node[below right] {0} (v01010);
			\draw [->] (v0101) -- node[below left]  {1} (v01011);
			\draw [->] (v0110) -- node[below right] {0} (v01100);
			\draw [->,red, line width=0.8mm] (v0110) -- node[below left=-0.4mm and -0.4mm]  {1} (v01101);
			\draw [->] (v0111) -- node[below right] {0} (v01110);
			\draw [->] (v0111) -- node[below left]  {1} (v01111);
			
			\node[anchor=south] at (-2.4,-8.5) (lcafloor) {$\lfloor{\mathbf{lca}(x,y)}\rfloor$};
			\draw (lcafloor) edge[out=90,in=-90,->] (v01011);
			
			\node[anchor=south] at (0,-8.5) (lcaceil) {$\lceil{\mathbf{lca}(x,y)}\rceil$};
			\draw (lcaceil) edge[out=90,in=-90,->] (v01100);
			
			\node[anchor=east] at (-5,-7.5) (x) {$x$};
			\draw (x) edge[out=0,in=180,->] (v01001);
			
			\node[anchor=west] at (1.6,-7.5) (y) {$y$};
			\draw (y) edge[out=180,in=0,->] (v01101);
			
			\node[anchor=west] at (0,-1.7) (lca) {$\mathbf{lca}(x,y)$};
			\draw (lca) edge[out=180,in=0,->] (v01);
			
			\draw [line width=0.2mm] (v00) -- (-5,-4.6) -- (-6.8,-4.6) -- (v00);
			\draw [line width=0.2mm] (v10) -- (2.4,-4.6) -- (4.2,-4.6) -- (v10);
			\draw [line width=0.2mm] (v11) -- (4.3,-4.6) -- (6.1,-4.6) -- (v11);
			
			\draw [decorate,decoration={brace,amplitude=15pt}] (6.2,0.2) -- (6.2,-6.9) node [black,midway,xshift=1.4cm] {height=$\gamma$};
			
			\end{tikzpicture}
			\caption{A complete binary tree of height $\gamma$ having $k=2^\gamma$ leaves. In this example, $\gamma=5$, $x=01001$ and $y=01101$. Thus, $\mathtt{Ham}(x,y)=1$ and $\lvert\mathbf{lca}(x,y)\rvert=2$.}
			\label{fig:ptree}
		\end{center}
	\end{figure}

	\subsection{Preliminaries}
	
	We first set up some terminology that we use in this section. Let $k =
	2^{\gamma}$, where $\gamma$ is a positive integer. We identify the
	numbers in the set $\{0,1,\ldots,k-1\}$ with elements of
	$\{0,1\}^\gamma$ using their $\gamma$-bit binary representation.  We
	index the bits of the binary strings from left to right using integers
	$i=1,2,\ldots,\gamma$. Thus, $x[i]$ denotes the $i$-th bit of $x$ (from the
	left); we use $x[i,j]$ to denote the string $x[i]\,x[i+1]\ldots\,
	x[j]$ of length $j-i+1$ (here $i,j$ satisfy $1 \leq i \leq j
	\leq \gamma$).
	
	For a string of $\gamma$ bits $a$, we use $\rev(a)$ to represent the \emph{reverse} of
	$a$; that is, the binary string obtained by writing the bits of $a$ in reverse (for instance, $\mathbf{rev}_5(00010) = 01000$). We may arrange binary
	strings in a binary tree. Refer to~\autoref{fig:ptree} for an
	example. The root is the empty string; the left child of a vertex
	$x$ is the vertex $x \, 0$, and its right child is the vertex $x\,1$.
	In particular, the string $y$ is a \emph{descendant} of the string $x$
	if $y$ is obtained by concatenating $x$ with some (possibly empty) string $z$; that is, $y = x\,z$. Consider the binary tree of depth
	$\gamma$, whose leaves correspond to elements of $\{0,1\}^\gamma$. For
	distinct elements $x,y \in \{0,1\}^\gamma$, let $\lca(x,y)$ be the
	\emph{lowest common ancestor} of $x$ and $y$ defined as follows:
	$$\lca(x,y) = x[1,\ell-1] = y[1,\ell-1], \text{ where }
	\ell=\min\left\{i\in[\gamma]: x[i] \neq y[i]\right\}.$$ For example,
	$\lca(0100111,0101010)= 010$.  Let $\lfloor{\lca(x,y)}\rfloor$ be the
	\emph{floor of $\lca(x,y)$}, and $\lceil{\lca(x,y)}\rceil$ be the
	\emph{ceiling of $\lca(x,y)$} defined as follows:
	\begin{align*}
	\lfloor{\lca(x,y)}\rfloor &= \lca(x,y)\, 0\, 1^{\gamma-\ell}\\
	\lceil{\lca(x,y)}\rceil &= \lca(x,y)\, 1\, 0^{\gamma-\ell}
	\end{align*}
	Since $\lfloor{\lca(x,y)}\rfloor,\lceil{\lca(x,y)}\rceil\in
	\{0,1\}^\gamma$, we may regard $\lfloor{\lca(x,y)}\rfloor$ and
	$\lceil{\lca(x,y)}\rceil$ as numbers in the set $\{0,1,\ldots,k-1\}$. Note
	that $\lfloor \lca(x,y) \rfloor =\lceil \lca(x,y) \rceil -1$, and if
	$x<y$, then $\lfloor{\lca(x,y)}\rfloor \in [x,y)$ and
	$\lceil{\lca(x,y)}\rceil\in
	(x,y]$\footnote{$[x,y]\triangleq\{x,x+1,x+2,\ldots,y\}$ and
		$[x,y)\triangleq\{x,x+1,x+2,\ldots,y-1\}$.}.

	Strings in $\{0,1\}^\gamma$ can also be viewed as vertices of an
	$\gamma$-dimensional hypercube, with edge set
	\[ \cH_\gamma = \{(x,x'): x,x' \in \{0,1\}^\gamma \text{ and } 
	x < x' \text{ and } \ham(x,x') =1 \},\] where $\ham(x,x')$ is the
	Hamming distance between $x$ and $x'$. Thus, if $(x,x') \in
	\cH_\gamma$, then $x$ and $x'$ differ at a unique location where $x$
	has a zero and $x'$ a one.
	\begin{Claim} \label{partsclaim}
		Suppose $(x,x')$ and $(y,y')$ are distinct edges of $\cH_\gamma$. 
		\begin{enumerate}
			\item[(a)] If $\lca(x,x') = \lca(y,y')$, then $[\rev(x),\rev(x')] \cap [\rev(y), \rev(y')] = \emptyset$.
			\item[(b)] If $\{\floor{\lca(x,x')}, \floor{\lca(y,y')}\} \subseteq
			[x,x') \cap [y,y')$, then
			$$[\rev(x),\rev(x')]\cap [\rev(y), \rev(y')] = \emptyset.$$
		\end{enumerate}
	\end{Claim}
	\begin{proof}
		Although part (b) implies part (a), it is easier to show part (a)
		first, and then derive part (b) from it. For part (a), let
		$|\lca(x,x')| = |\lca(y,y')| = \ell-1$. Let $a,b \in \{0,1\}^{\gamma -
			\ell}$ be such that 
		\[
		a = x[\ell+1,\gamma] = x'[\ell+1,\gamma] \neq y[\ell+1,\gamma] =
		y'[\ell+1,\gamma] = b.
		\]
		In particular, we have $a\neq b$ (implying
		$\revl(a)\neq\revl(b)$). Note that $\rev(a)$ represents the
		$\gamma-\ell$ most significant bits of $\rev(x)$ and $\rev(x')$;
		similarly, $\rev(b)$ represents the $\gamma-\ell$ most significant
		bits of $\rev(y)$ and $\rev(y')$.

		If $\revl(a)<\revl(b)$ then $\rev(x')<\rev(y)$; and if $\revl(b)<\revl(a)$ then $\rev(y')<\rev(x)$. In either case, $[\rev(x),\rev(x')]$ and $[\rev(y),\rev(y')]$ are disjoint, proving part (a).
		
		Next, consider part (b). Suppose $\floor{\lca(x,x)},\floor{\lca(y,y')}
		\in [x,x') \cap [y,y')$. Since every $p\in[x, x')$ is a descendant of
		$\lca(x,x')$, we conclude that $\lca(y,y')$ is a descendant of
		$\lca(x,x')$. Similarly, $\lca(x,x')$ is a descendant of
		$\lca(y,y')$. But then $\lca(x,x')=\lca(y,y')$, and part (b)
		follows from part (a).
	\end{proof}
	
	\subsection{Manhattan Graphs}
	
	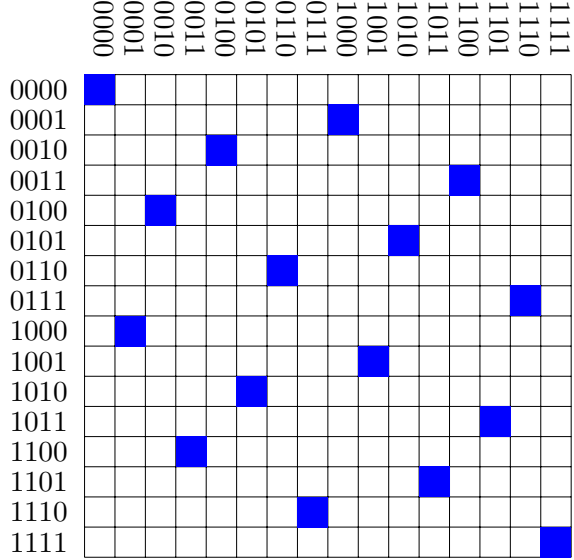
\begin{figure}
		\begin{center}
			\begin{tikzpicture}
			\draw[step=.4, ultra thin] (0,0) grid (6.4,6.4);
			
			\foreach \a in {0,1}
			\foreach \b in {0,1}
			\foreach \c in {0,1}
			\foreach \d in {0,1}
			{
				\node at (-.6,6.2-.4*\a-.8*\b-1.6*\c-3.2*\d)
				{\small \d\c\b\a};
				
				\fill[fill=blue, opacity=1](.4*\a+.8*\b+1.6*\c+3.2*\d,
				6.4-.4*\d-.8*\c-1.6*\b-3.2*\a)
				rectangle +(.4,-.4);
				
				\node[rotate=270] at (.25+.4*\a+.8*\b+1.6*\c+3.2*\d,7)
				{\small \d\c\b\a};
			}
			
			\end{tikzpicture}
			\caption{The bit-reversal permutation matrix for $k=16$. In $\rect$, each cell of this matrix represents a vertex. The blue cells represent the terminal vertices of $\Tmid$; all the other vertices are non-terminals. Edges are named horizontal, upward and downward in the natural way.}
			\label{fig:rect}
		\end{center}
	\end{figure}
	
	In this section, we describe a directed grid graph $\rect$ (which we refer
	to as the Manhattan graph) with $3k$ terminals. We show that any
	distance-preserving subgraph of $\rect$ has $\Omega(k\log k)$
	branching vertices. The graph has $k^2+2k$ vertices arranged in
	a square grid. The vertices and edges of $\rect$ are defined as follows. (\autoref{fig:rect} makes this definition easier to understand.)
	\begin{enumerate}
		\item $V(\rect)=\{0,1,2,\ldots,k-1\}\times\{-1,0,1,\ldots,k\}$.
		\item There are three kinds of edges: horizontal,
		upward and downward; the edge set is given by
		$E(\rect)=\hor\cup\upvert\cup\downvert$, where
		\begin{align*}
		\hor &= \{ ((i,j),(i,j+1)): i =0,1,\ldots,k-1 \text{ and }
		j =-1,0,\ldots,k-1 \};\\
		\upvert &= \{ ((i_1,j),(i_2,j)): 0\leq i_2<i_1\leq k-1 \text{ and }
		j =-1,0,\ldots,k\};\\
		\downvert &= \{ ((i_1,j),(i_2,j)): 0\leq i_1<i_2\leq k-1 \text{ and }
		j =-1,0,\ldots,k\}.
		\end{align*}
		
		\item The edge weights are given by the function
		$w:E(\rect)\rightarrow\{0,1\}$, defined as follows:
		$w(e) = 1$ if $e \in \hor \cup \upvert$, and $w(e)=0$ if $e \in \downvert$.
	\end{enumerate}
	The set of terminals are of the form $T=T_{\text{left}}\cup\Tmid\cup
	T_{\text{right}}$, where
	\begin{align*}
	T_{\text{left}} &= \{0,1,\ldots,k-1\} \times \{-1\}, \\
	T_{\text{right}} &= \{0,1,\ldots,k-1\} \times \{k\}; \\
	T_{\text{mid}} &= \{ (\rev(i),i): i=0,1,\ldots,k-1\}.
	\end{align*}
	This completes the definition of $\rect$.

	Fix an optimal distance-preserving subgraph $\recth$ of $\rect$. We
	shall show that $\recth$ has $\Omega(k\log k)$ vertices of degree at
	least $3$.
	\begin{Lemma}\label{degtwo}
		$V(\recth)=V(\rect)$ and $\hor \subseteq E(\recth)$.
	\end{Lemma}
	\begin{proof}
		For any $i\in\{0,1,\ldots,k-1\}$, note that the \emph{unique} shortest path between the terminals $(i,-1)$ and
		$(i,k)$ is precisely $((i,-1),(i,0),\ldots,(i,k))$. Thus, all vertices and all
		horizontal edges in the $i$-th row of $\rect$ must be part of $\recth$.
		%
	\end{proof}
	It follows from~\autoref{degtwo} that every non-terminal vertex in $\recth$ has degree at least two, namely the two horizontal edges incident on it.
	
	From now on, we rely solely on the fact that $\recth$ is distance-preserving for every pair of terminals in $\Tmid$; that is, we prove the stronger statement that just preserving terminal distances in $\Tmid$ requires $\Omega(k\log k)$ branching vertices.
	
	Order the vertices
	in $\Tmid$ as $t_0,t_1,\ldots,t_{k-1}$, where $t_i=(\rev(i),i)$. Note that these terminals appear in different rows
	and columns. Consider the following pairs of terminals. (We call such pairs ``friends".)
	\[
	\Friends = \{(t_i,t_j): (i,j) \in \cH_\gamma \}.
	\]
	For each pair of friends $(t_i,t_j)$, fix $P(i,j)$, a path of minimum distance between $t_i$ and $t_j$
	in $\recth$. We are now set to formally define \emph{special edges}.
	
	\begin{Definition}
		Let $\spec(i,j)=((r_{ij},\floor{\lca(i,j)}), (r_{ij},\lceil\lca(i,j)\rceil))$ be an edge of $P(i,j)$, where $\rev(i) \leq r_{ij} \leq \rev(j)$. (By~\autoref{horizonlemma}, such an edge exists.) Let $\spec = \{\spec(i,j): (t_i,t_j) \in \Friends\}$.
	\end{Definition}
	
	\begin{Lemma} \label{horizonlemma} Let $(t_i,t_j)\in\Friends, \ell
		= \floor{\lca(i,j)}$. Then, there is an $r_{ij}\in[\rev(i),\rev(j)]$ such that $P(i,j)$ contains the edge
		$((r_{ij}, \ell),(r_{ij},\ell+1))$.
	\end{Lemma}
	
	\begin{proof}
		We have $i<j$, $t_i=(\rev(i),i)$ and $t_j=(\rev(j),j)$. Also note that since $(i,j)\in\cH_\gamma$, $\rev(i)<\rev(j)$. Thus, $d(i,j) = j-i$, and the
		shortest path $P(t_i,t_j)$ goes from column $i$ to column $j$ and
		never skips a column.  Since $\ell \in [i,j)$, there must be an edge
		in $P(i,j)$ of the form $((r_{ij}, \ell), (r_{ij}, \ell+1))$ (say,
		the edge of $P(i,j)$ that leaves column $\ell$ for the last
		time). We claim that $r_{ij} \in [\rev(i),\rev(j)]$. For otherwise,
		$P(i,j)$ would contain an edge in $\upvert$. Then, apart from the
		$j-i$ edges from $\hor$, $P(i,j)$ would contain an additional edge
		from $\upvert$ of weight $1$; that is, the length of $P(i,j)$ would
		be at least $j-i+1$---contradicting the fact that $d(i,j) = j -i$.
	\end{proof}
	
	\begin{Lemma}[Key lemma] \label{kilemma}
		Suppose $(t_x,t_{x'})$ and $(t_y,t_{y'})$ are distinct pairs in
		$\Friends$ such that their special edges are in the same row $r$, that
		is,
		\begin{align*}
		\spec(x,x')&=
		((r,\alpha),(r, \alpha + 1))\\ 
		\spec(y,y')&=((r,\beta),(r,\beta+ 1)),
		\end{align*}
		where $\alpha = \floor{\lca(x,x')}$ and $\beta = \floor{\lca(y,y')}$. 
		\begin{enumerate}
			\item[(a)] Then, $\alpha \neq \beta$. In particular, $\spec(x,x')\neq \spec(y,y')$.
			
			\item[(b)] Suppose $\alpha < \beta$. Then, there exists an $\ell \in
			[\alpha+1, \beta]$ such that $(r,\ell)$ is either a branching vertex
			or a terminal in $\recth$.
		\end{enumerate}
	\end{Lemma}
	\begin{proof}
		Part (a) follows from~\autoref{partsclaim} (a). Consider part (b).  By
		our definition of special edge, $r\in[\rev(x),\rev(x')]$ and
		$r\in[\rev(y),\rev(y')]$. So, $[\rev(x),\rev(x')]\cap [\rev(y),
		\rev(y')] \neq \emptyset$, and by~\autoref{partsclaim} (b) (in the
		contrapositive) either $\alpha \notin [y,y')$ or $\beta\notin [x,x')$.
		If $\alpha \notin [y,y')$, $\spec(x,x')$ is not on
		$P(y,{y'})$. The first vertex in row $r$ that is part of
		$P(y,{y'})$ is in a column $\ell \in [\alpha+1,\beta]$.
		Then, $(r,\ell)$ is either a branching vertex or the terminal
		$t_y$. On the other hand, if $\beta\notin [x,x')$, then the last
		vertex of $P(t_x,t_{x'})$ in row $r$ lies in a column $\ell
		\in [\alpha+1,\beta]$, so $(r,\ell)$ is either a branching
		vertex or the terminal $t_{x'}$.
	\end{proof}
	\begin{Corollary}[Corollaries of~\autoref{kilemma}] \label{cor:rect} \ 
		\begin{enumerate}
			\item[(a)] $|\spec| = |\Friends| = |\cH_\gamma| = k \log k/2$.
			\item[(b)] If two edges in $\spec$ fall in the same row, then there is
			a branching vertex or a terminal separating them.
		\end{enumerate}
	\end{Corollary}
	\begin{Theorem} \label{thm:manhattan}
		$\recth$ has $\Omega(k\log k)$ branching vertices.
	\end{Theorem}
	\begin{proof}
		For each $i\in\{0,1,\ldots,k-1\}$, let $\delta_i$ be the number of
		distinct edges in $\spec$ in row $i$. 
		Then, by~\autoref{cor:rect} (a), we have
		$$\displaystyle{\sum\limits_{i=0}^{k-1}\delta_i=
			|\spec| = \left(\frac{k\log
				k }{2}\right)}.$$
		Furthermore,~\autoref{cor:rect} (b) implies that there are
		at least $\delta_i-2$ many branching vertices of the form $(i,x)$ in
		$\recth$, where $0\leq x\leq k-1$. Thus, the total number of branching
		vertices in $\recth$ is at least
		$$\displaystyle{(\delta_0-2)+(\delta_1-2)+\cdots+(\delta_{k-1}-2)=\left(\sum\limits_{i=0}^{k-1}\delta_i\right)-2k=\left(\frac{k\log k}{2}\right)-2k}.$$
		Since this quantity is $\Omega(k\log k)$, this completes the proof.
	\end{proof}

	\subsection{Translating the Lower Bound to Interval Graphs}

	In this section, we present an interval graph $\intgraph$ with $O(k)$ terminals, for which every distance-preserving subgraph has $\Omega(k \log k)$ branching vertices. Our lower bound relies on the lower bound for the Manhattan graph shown in the previous section. \autoref{fig:transform} can be helpful to navigate through this proof. Let us describe the interval graph. Let $\cJ$ be the set of intervals.
	\[ 
	\cJ = \{[x,x+1]: x = -1, -1+ 1/k, \ldots,-1/k, 0, \ldots, k, k + 1/k,
	\ldots, k +1 -1/k \}.\] Thus, we have unit intervals starting at all
	integral multiples of $1/k$ in the range $[-1,k+1-1/k]$; in all we
	have $k(k+2)$ intervals in $\cJ$. These intervals naturally define an interval graph. Furthermore, the edges of $\intgraph$ are \emph{directed} as follows. Orient the edges of $\intgraph$ \emph{from} an earlier interval \emph{to} a later interval; that is, $([x,x+1],[y,y+1])$ is a directed edge \emph{from} $[x,x+1]$ \emph{to} $[y,y+1]$ if and only if $x<y\leq x+1$. Note that this orientation does not affect shortest paths. Any shortest path from $[i,i+1]$ to $[j,j+1]$ (where $i<j$) in the undirected interval graph is also a valid directed shortest path in $\intgraph$. Also, $\intgraph$ has $k^2+2k$ vertices, which (surprisingly?) is the number
	of vertices in the Manhattan graph of the previous section. In fact, the
	connection is deeper. Let us arrange the intervals in a two-dimensional array
	\[ \mathbf{A} = \langle a_{i,j}: i=0,\ldots,k-1 \text{ and } j= -1,0,\ldots,k \rangle,\]
	where $a_{ij}$ corresponds to the interval $[j+(k-1-i)/k,
	j+1+(k-1-i)/k]$. Thus, the first $k$ intervals of $\cJ$ occupy the
	left most column of the array $\mathbf{A}$ (from bottom to top); the
	next $k$ intervals occupy the next column (again from bottom to top),
	and so on. It is easy to check that, after this arrangement, the directed edges of $\intgraph$ are of three types: horizontal, upward and slanting.
	\begin{align*}
	\hor(\intgraph)
	&= \{(a_{i,j},a_{i,j+1}): 0\leq i\leq k-1 \text{ and } -1\leq j\leq k-1\};\\
	\upvert(\intgraph) &= 
	\{(a_{i,j},a_{i',j}): 1\leq i\leq k-1\
	\text{ and }
	0\leq i'<i\ 
	\text{ and }
	-1\leq j\leq k\
	\};\\
	\Eslant(\intgraph) &= \{(a_{i,j},a_{i',j+1}):
	0\leq i\leq k-2 \text{ and } i<i'\leq k-1
	\text{ and }
	-1\leq j\leq k-1\}.
	\end{align*}
	Thus, $E(\intgraph) = \hor(\intgraph)\cup \upvert(\intgraph) \cup \Eslant(\intgraph)$. All edges in $E(\intgraph)$ have weight $1$. This 2d array can be viewed as a $k\times (k+2)$ grid, and we place terminals in this graph at the same $3k$ locations as in the Manhattan graph. This completes the description of $\intgraph$.

	
	


	
	Let $\intgraphh$ be a distance-preserving subgraph of $\intgraph$. Note that the terminals in the first and last column ensure that all horizontal edges must be part of $\intgraphh$. So, both end points of every slanting edge and every upward edge included in $\intgraphh$ are branching vertices. Our proof strategy is as follows. We obtain from $\intgraphh$ a distance-preserving subgraph $\newgraph$ of the
	Manhattan graph with \emph{nearly} the same number of branching vertices. Since $\newgraph$ requires $\Omega(k \log k)$ branching vertices, the number of branching vertices in $\intgraph$ is $\Omega(k \log k)$.
	
	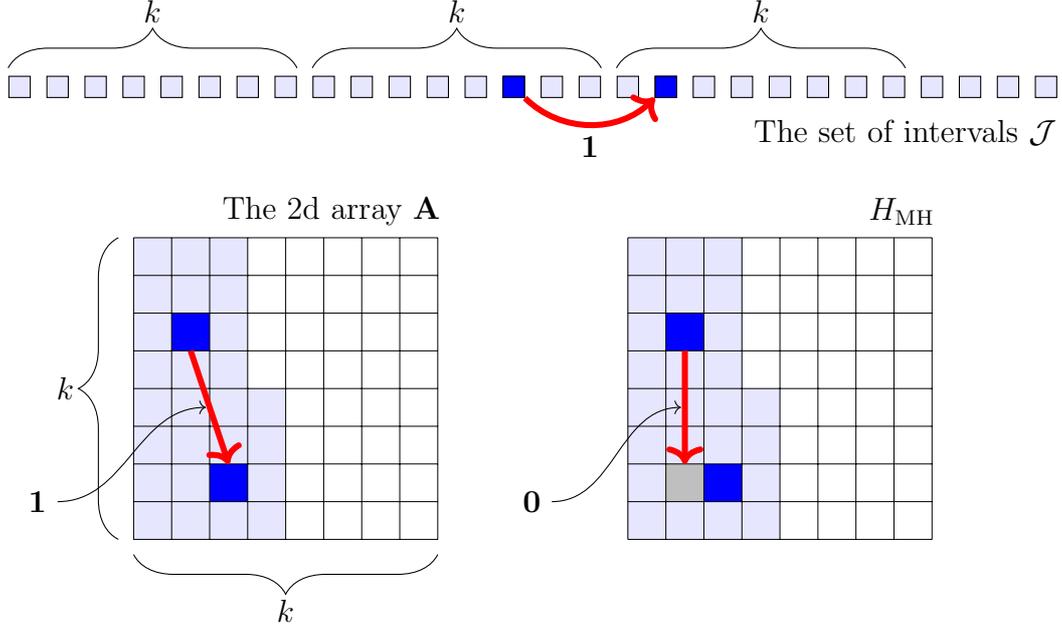
\begin{figure}
		\begin{center}
			\begin{tikzpicture}
			
			\def \o {1.5};
			\def \y {6};
			
			\foreach \c in {0,...,27}
			{
				\def \i {0.5*\c};
				\vertbox at (\i-\o,\y) [fill=blue!10] {};
			}
			
			\vertbox(u1) at (6.5-\o,\y) [fill=blue] {};
			\vertbox(v1) at (8.5-\o,\y) [fill=blue] {};
			
			\node at (7.5-\o,\y-0.8) (y) {$\bf 1$};
			\node at (11.65-\o,\y-0.6) (y) {The set of intervals $\cJ$};
			\node at (2.6,4.35) {The 2d array $\mathbf{A}$};
			\node at (10.1,4.35) {$\newgraph$};
			
			\draw [decorate,decoration={brace,amplitude=15pt}] (-0.15-\o,0.25+\y) -- (3.65-\o,0.25+\y) node [black,midway,yshift=0.75cm] {$k$};
			\draw [decorate,decoration={brace,amplitude=15pt}] (3.85-\o,0.25+\y) -- (7.65-\o,0.25+\y) node [black,midway,yshift=0.75cm] {$k$};
			\draw [decorate,decoration={brace,amplitude=15pt}] (7.85-\o,0.25+\y) -- (11.65-\o,0.25+\y) node [black,midway,yshift=0.75cm] {$k$};
			
			\draw [decorate,decoration={brace,amplitude=15pt}] (-0.2,0)--(-0.2,4) node [black,midway,xshift=-0.7cm] {$k$};
			\draw [decorate,decoration={brace,amplitude=15pt}] (4,-0.2)--(0,-0.2) node [black,midway,yshift=-0.75cm] {$k$};
			
			\foreach \c in {0,...,7}
			\foreach \j in {0,0.5,1}
			{
				\def \i {0.5*\c};
				\fill[fill=blue!10](\j,\i) rectangle +(0.5,0.5);
				\fill[fill=blue!10](\j+6.5,\i) rectangle +(0.5,0.5);
			}
			
			\foreach \c in {0,...,3}
			{
				\def \i {0.5*\c};
				\fill[fill=blue!10](1.5,\i) rectangle +(0.5,0.5);
				\fill[fill=blue!10](8,\i) rectangle +(0.5,0.5);
			}
			
			\fill[fill=lightgray](7,0.5) rectangle +(0.5,0.5);
			
			\fill[fill=blue, opacity=1](0.5,2.5) rectangle +(0.5,0.5);
			\fill[fill=blue, opacity=1](1,0.5) rectangle +(0.5,0.5);
			
			\fill[fill=blue, opacity=1](7,2.5) rectangle +(0.5,0.5);
			\fill[fill=blue, opacity=1](7.5,0.5) rectangle +(0.5,0.5);
			
			\draw[step=.5, ultra thin] (0,0) grid (4,4);
			\draw[step=.5, ultra thin] (6.5,0) grid (10.5,4);
			\draw(6.5,0)--(6.5,4);
			
			\draw [->,red, line width=0.8mm] (u1) to [out=-45,in=225](v1);
			\draw [->,red, line width=0.8mm, label=below:1] (0.75,2.5) to (1.25,1);
			\draw [->,red, line width=0.8mm] (7.25,2.5) to (7.25,1);

			\node[anchor=east] at (-1,0.5) (x) {$\bf 1$};
			\draw (x) edge[out=0,in=180,->] (0.95,1.75);
			
			\node[anchor=east] at (5.5,0.5) (y) {$\bf 0$};
			\draw (y) edge[out=0,in=180,->] (7.2,1.75);
			
			\end{tikzpicture}
			\caption{The transformation: (i) The set of intervals $\cJ$ (represented by their starting points) is divided into groups of size $k$ each. (ii) Then, each group is placed in a column of the 2d array $\mathbf{A}$ from bottom to top. (iii) Finally, each slanting edge (weight 1) is replaced by a downward edge (weight 0) to obtain $\newgraph$. Note that the distance between the pair of blue vertices is 1 in all three graphs. In $\cJ$ and $\mathbf{A}$, they are connected by a single edge of weight 1. In $\newgraph$, the gray vertex has a weight 1 edge to the blue vertex in its adjacent column.}
			\label{fig:transform}
		\end{center}
	\end{figure}

	\emph{The transformation (\autoref{fig:transform}):} We retain all corresponding vertices and all upward and horizontal edges of $\intgraphh$ in $\newgraph$. Now, $\intgraphh$ might include slanting edges of the form $(a_{i,j},a_{i',j+1})$ (where $i<i'$), but the pair $p=((i,j),(i',j+1)) \not\in E(\rect)$. So, we accommodate such slanting edges in $\newgraph$ by providing a path of weight $1$ between its end points. We replace each slanting edge $(a_{i,j},a_{i',j+1})$ in $E(\intgraphh)$ with the 0-weight edge $((i,j),(i',j))$ (a downward edge of weight zero) in $E(\newgraph)$. Note, that the edge $((i',j), (i',j+1))$ is a horizontal edge and is already retained in $E(\newgraph)$. Thus, $\newgraph$ has a path of total weight $1$, namely $(i,j)\rightarrow(i',j)\rightarrow(i',j+1)$, that connects the end points of the pair $p$. Let $\hor(\intgraphh)$, $\upvert(\intgraphh)$ and $\Eslant(\intgraphh)$ be the horizontal, upward and slanting edges of $\intgraph$ that are part of $\intgraphh$. Then, $E(\newgraph) = \hor(\newgraph) \cup \upvert(\newgraph) \cup \downvert(\newgraph)$, where
	\begin{align}
	\hor(\newgraph) & = \{((i,j),(i,j+1)): (a_{i,j}, a_{i,j+1}) \in \hor(\intgraphh)\}; \\
	\upvert(\newgraph) & = \{((i_1,j),(i_2,j)): (a_{i_1,j}, a_{i_2,j}) 
	\in \upvert(\intgraphh)\};\\
	\downvert(\newgraph) & = \{((i_1,j),(i_2,j)): (a_{i_1,j}, a_{i_2,j+1}) 
	\in \Eslant(\intgraphh)\}. \label{replace:sllant}
	\end{align}
	
	It is straightforward to verify that $\newgraph$ preserves distances
	between all pairs of terminals in $\rect$. However, for each slanting
	edge we replace, we might create a new branching vertex (for example,
	the vertex $(i_2,j)$ created in~\autoref{replace:sllant} might be a
	branching vertex in $\newgraph$ with no corresponding branching vertex
	in $\intgraphh$). The number of such vertices is at most the number of slanting edges, which in turn is at most the number of branching vertices in $\intgraphh$. Thus, the total number of branching vertices
	in $\newgraph$ is at most twice the number of branching vertices in
	$\intgraphh$ (plus $O(k)$ to account for downward edges in the last column). Using~\autoref{thm:manhattan}, the number of
	branching vertices in $\intgraphh$ is $\Omega(k \log k)$, completing the proof of~\autoref{thm:lb}.

	\section{Branching Vertices versus Branching Edges}

	\begin{figure}
		\begin{center}
			\begin{tikzpicture}
			
			\def \d {0.2};
			\def \l {5};
			\def \p {-0.4}
			\foreach \c in {-5,...,0}
			{
				\def \b {\c+6};
				\def \a {\c/2+3};
				\draw [blue!40, line width=0.7mm] (\b,\a-\d)--(\b,\a)--(\b+\l,\a)--(\b+\l,\a-\d);
				\draw [blue!40, line width=0.7mm] (\b,\a+\d)--(\b,\a)--(\b+\l,\a)--(\b+\l,\a+\d);
				\vertex at (\b,\p) [fill=gray, label=below:\c] {};
			}
			\foreach \c in {1,...,5}
			\vertex at (\c+6,\p) [fill=gray, label=below:\c] {};
			
			\def \r {4};
			\def \s {\r/2};
			
			\draw [red, line width=1.2mm] (\r,\s-\d)--(\r,\s)--(\r+\l,\s)--(\r+\l,\s-\d);
			\draw [red, line width=1.2mm] (\r,\s+\d)--(\r,\s)--(\r+\l,\s)--(\r+\l,\s+\d);
			
			\def \g {0.2+\p};
			\def \h {-0.8+\p};
			
			\draw [densely dashed, very thick, blue] plot [smooth cycle] coordinates {(\r,\g) (\r+\l,\g) (\r+\l,\h) (\r,\h)};


			\end{tikzpicture}
			\caption{The interval graph $\Gzero$ for $k=5$. Each non-terminal covers $k+1=6$ terminals.}
			\label{fig:hansel}
		\end{center}
	\end{figure}
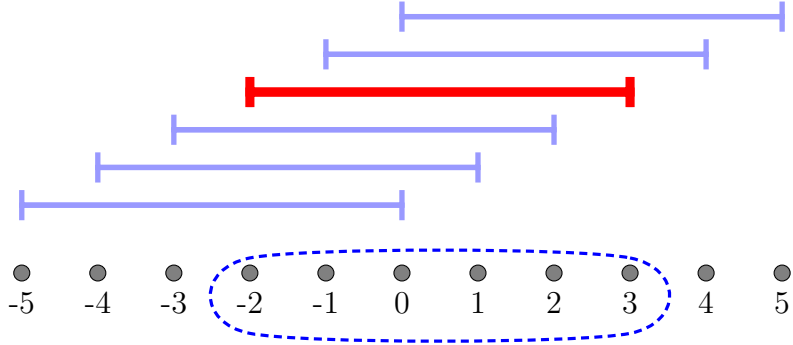
	
	In our formulation, we count the number of branching vertices (vertices with degree $\geq 3$). It is also reasonable to consider the number of edges incident on non-terminal branching vertices (we refer to such edges as \emph{branching edges}) as the measure of complexity. Our $\Omega(k\log k)$ lower bound (\autoref{thm:lb}) is clearly applicable to the number of branching edges as well.
	
	In this section, we show a separation between the number of branching vertices and the number of branching edges. In particular, we present an interval graph $\Gzero$ with $k$ terminals, each of length zero, such that the total number of branching edges in any distance-preserving subgraph of $\Gzero$ must be $\Omega(k \log k)$. However, $\Gzero$ admits a distance-preserving subgraph with $O(k)$ branching vertices.
	
	Let us now describe $\Gzero$. The interval representation of $\Gzero$ has $k+1$ non-terminals of unit length each, and $k$ terminals of zero length each. Let the intervals corresponding to non-terminal vertices of $\Gzero$ be
	\[ \{ [x,x+k]: x = -k,-k+1, \ldots,0\},\]
	Let the intervals corresponding to terminal vertices of $\Gzero$ be \[ \{ [x,x]: x=-k, -k +1, \ldots,0, 1, \ldots,k\}.\] See~\autoref{fig:hansel} for an instance of $\Gzero$.

	\begin{Theorem}
		Every distance-preserving subgraph of $\Gzero$ has at least $\Omega(k\log k)$
		branching edges and at most $O(k)$ branching vertices.
	\end{Theorem}
	\begin{proof}
	    Since the total number of vertices in $\Gzero$ is $O(k)$, every distance-preserving subgraph of $\Gzero$ has $O(k)$ branching vertices. Now we prove that every distance-preserving subgraph of $\Gzero$ has $\Omega(k\log k)$ branching edges.
		
		Fix a distance-preserving subgraph $H$ of $\Gzero$.
		Consider pairs of terminals in the set $\{t_x: x =-k,-k+1,\ldots,-1\}
		\times \{t_y: y=0,1,\ldots,k-1\}$, and restrict attention to those
		pairs that are at distance two in $\Gzero$; that is, pairs that are covered by a common interval in $\cI$. Indeed, for every pair of integers $i,j$ where
		$0\leq i<j \leq k$, the pair $(t_{j-k}, t_i)$ is at distance two in
		$\Gzero$. Build an auxiliary graph $\cP$ on the vertex set $\{1,2,\ldots,k\}$,
		where the pair $(i,j)$ is an edge if $(t_{j-k-1}, t_i)$ is at distance
		two.  Clearly, $\cP$ is a complete graph on $k$ vertices.  For every 
		interval $I \in \cI$, let $B_I$ be the subgraph of $\cP$ 
		with vertex set $\{1,2,\ldots,k\}$ and edge set 
		\[ E(B_I)= \{ (i,j): (t_{j-k-1},I)\in E(H)\text{ and } (t_i, I) \in E(H)\}.\]
		One can verify that for each $I\in\cI$, the graph $B_I$ is bipartite, and the number of
		\emph{non-isolated} vertices in $B_I$ is at most the degree of the
		vertex $I$ in $H$. By a result of Hansel~\cite{Hansel} stated below (see also~\cite{Katona,Pipp}), the total number of non-isolated vertices in $\bigcup_I B_I=\cP$ is at least $k \log
		k$. Thus, the total number of edges in $\Gzero$ is at least $k \log
		k$. Since $\Gzero$ has $O(k)$ vertices, at most $O(k)$ of these edges of $H$
		can be incident on vertices of degree at most two. It follows that $H$ has
		$\Omega(k\log k)$ branching edges.
	\end{proof}
	
	\begin{Lemma}[Hansel~\cite{Hansel}]
		Let $K_n$ be the complete graph on $n$ vertices, and let $B_1,B_2,\ldots,B_r$ be $r$ bipartite graphs on the vertex set $\{1,2,\ldots,n\}$, such that $\bigcup_i E(B_i) = E(K_n)$. Suppose the number of non-isolated vertices in $B_i$ is $n_i$. Then \[ n_1 + n_2 + \ldots + n_r \geq n \log n.\]
	\end{Lemma}

	
	

	\section{Conclusion}
	
	In this paper, we studied distance-preserving subgraphs and solved the problem conclusively for interval graphs (\autoref{thm:main}) by proving matching upper and lower bounds (up to constants). However, some interesting open questions still remain.
	
	Is there a polynomial time algorithm to find an \emph{optimal} distance-preserving subgraph of an interval graph? This problem is $\NP$-hard (\autoref{thm:npcnpc}), but is it fixed-parameter tractable (with parameter $k$, the number of terminals) for general graphs?
	
	It is also interesting to consider classes of graphs that are generalisations of interval graphs (perfect graphs, chordal graphs), and to check if our ideas can be extended to those classes as well.
	
	
	

	\subsection*{Acknowledgments}
	
	We are grateful to Nithin Varma and Rakesh Venkat for introducing us to the problem and helping with the initial analysis of shortest paths in interval graphs, and for their comments at various stages of this work. We would also like to thank the anonymous reviewers of this paper for their helpful suggestions and comments.

	
	

	\bibliographystyle{alpha}
	\bibliography{FirstPaper.bib}
	
\end{document}